\documentclass[runningheads,envcountsame]{llncs}

\usepackage{amsmath}

\usepackage{amsthm}
\usepackage{amssymb}
\usepackage{xcolor}

\newcommand{\shortrules}[6]{\noindent\begin{minipage}{#6ex}{\bfseries #1}\end{minipage} $\;$ #2 $\;\Rightarrow_{\text{#5}}\;$ #3 \par\smallskip\noindent #4}
\newcommand{\longnamerules}[5]{\noindent{\bfseries #1}\newline \hspace*{3ex} #2 $\;\Rightarrow_{\text{#5}}\;$ #3 \par\smallskip\noindent #4}

\newcommand{\opers}{\Omega}

\newcommand{\term}{T}

\newcommand{\varset}{\mathcal{X}}

\newcommand{\conv}{\ensuremath{\operatorname{conv}}}

\newcommand{\mpos}{\ensuremath{\operatorname{pos}}}

\newcommand{\comp}{\operatorname{comp}}

\newcommand{\core}{\operatorname{core}}
\newcommand{\cores}{\operatorname{cores}}

\newcommand{\gnd}{\operatorname{gnd}}
\newcommand{\mgu}{\operatorname{mgu}}

\newcommand{\orewrites}[1]{\mbox{\ }^{#1}\hspace*{-1.5ex}\leftarrow}

\newcommand{\SCL}{{\text{SCL}}}

\newcommand{\SCLEQ}{{\text{\SCL(EQ)}}}

\newcommand{\clsr}[0]{\,\cdotp}

\mathchardef\mhyphen="2D % Define a "math hyphen"

\begin{document}
\title{SCL(EQ): SCL for First-Order Logic with Equality}
\author{Hendrik Leidinger\inst{1,2} \and Christoph Weidenbach\inst{1}}
\authorrunning{Leidinger et al.}
\institute{Max-Planck Institute for Informatics, Saarbr\"ucken, Germany\\
\email{\{hleiding, weidenbach\}@mpi-inf.mpg.de} \and
Graduate School of Computer Science, Saarbr\"ucken, Germany}
\maketitle              % typeset the header of the contribution
\begin{abstract}
  We propose a new calculus \SCLEQ\ for first-order logic with equality that only learns non-redundant clauses.
  Following the idea of CDCL (Conflict Driven Clause Learning) and SCL (Clause Learning from Simple Models) a ground
  literal model assumption is used to guide inferences that are then guaranteed to be non-redundant.
  Redundancy is defined with respect to a dynamically changing ordering derived from the ground literal model assumption.
  We prove \SCLEQ\ sound and complete and
  provide examples where our calculus improves on superposition.
%  There exist many complete procedures for first-order logic with equality. However, to the best of our knowledge, no approach grants both the learning of non-redundant clauses and term rewriting.
%  To enable learning of non-redundant clauses, we allow a dynamic ordering generated from a ground partial model assumption. Unlike the superposition calculus, this ordering changes during construction.
%  Computations with respect to the current model assumption are always effective because they are applied at any time to a finite set of ground terms. Yet, inferences are always applied to
%  the more general non-ground clauses. In summary, our calculus can be considered as a generalization of Knuth-Bendix Completion and Conflict driven clause learning (CDCL).
\keywords{First-Order Logic with Equality  \and Term Rewriting \and Model-Based Reasoning.}
\end{abstract}
\section{Introduction}

There has been extensive research on sound and complete calculi for first-order logic with equality. The current prime calculus is superposition~\cite{bachmair1994rewrite},
where ordering restrictions guide paramodulation inferences and an abstract redundancy notion enables a number of clause simplification and deletion mechanisms, such as rewriting
or subsumption. Still this ``syntactic'' form of superposition infers many redundant clauses. The completeness proof of superposition provides a ``semantic'' way of generating only
non-redundant clauses, however, the underlying ground model assumption cannot be effectively computed in general~\cite{Teucke18}.
It requires an ordered
enumeration of infinitely many ground instances of the given clause set, in general.
Our calculus overcomes this issue by providing an effective
way of generating ground model assumptions that then guarantee non-redundant inferences on the original clauses with variables.

The underlying ordering is based on the order of ground literals in the model assumption, hence changes during a run of the calculus. It incorporates a standard
rewrite ordering. For practical redundancy criteria this means that both rewriting and redundancy notions that are based on literal
subset relations are permitted
to dynamically simplify or eliminate clauses. Newly generated clauses are non-redundant, so redundancy tests are only needed backwards.
Furthermore, the ordering is automatically
generated by the structure of the clause set. Instead of a fixed ordering as done in the superposition case, the calculus finds and changes an ordering according
to the currently easiest way to make progress, analogous to CDCL (Conflict Driven Clause Learning)~\cite{MSS96,BayardoSchrag96,MoskewiczMadiganZhaoZhangMalik01,BiereEtAl09handbook,Weidenbach15}.

Typical for CDCL and SCL (Clause Learning from Simple Models)~\cite{AlagiWeidenbach15,FioriWeidenbach19,BrombergerFW21}  approaches to reasoning,
the development of a model assumption is done by decisions and propagations. A decision guesses a ground literal
to be true whereas a propagation concludes the truth of a ground literal through an otherwise false clause. While propagations in CDCL and propositional logic are
restricted to the finite number of propositional variables, in first-order logic there can already be infinite propagation sequences~\cite{FioriWeidenbach19}. In order to overcome
this issue, model assumptions in \SCLEQ\ are at any point in time restricted to a finite number of ground literals, hence to a finite number of ground instances of the clause set at hand.
Therefore, without increasing the number of considered ground literals, the calculus either finds a refutation or runs
into a \emph{stuck state} where the current model assumption satisfies the finite number of ground instances. In this case one
can check whether the model assumption can be generalized to a model assumption of the overall clause set or the information of the stuck
state can be used to appropriately increase the number of considered ground literals and
continue search for a refutation. \SCLEQ\ does not require exhaustive propagation, in general, it just forbids the decision of the complement of a literal
that could otherwise be propagated.

For an example of \SCLEQ\ inferring clauses, consider the three first-order clauses
\begin{center}
  $\begin{array}{c}
     C_1 := h(x) \approx g(x)\lor c\approx d \qquad C_2 := f(x) \approx g(x) \lor a \approx b\\
     C_3 := f(x) \not\approx h(x) \lor f(x) \not\approx g(x)\\
     \end{array}$
   \end{center}
 with a Knuth-Bendix Ordering (KBO), unique weight $1$, and  precedence $d\prec c\prec b\prec a \prec g \prec h \prec f$. A Superposition Left~\cite{bachmair1994rewrite}
 inference between $C_2$ and $C_3$ results in
\begin{center}
 $C'_4 :=  h(x) \not\approx g(x) \lor f(x) \not\approx g(x) \lor a \approx b$.
\end{center}
 For \SCLEQ~we start
 by building a partial model assumption, called a \emph{trail},  with two decisions
\begin{center}
$\Gamma := [h(a)\approx g(a)^{1:(h(x) \approx g(x)\lor h(x) \not\approx g(x))\cdotp \sigma},f(a)\approx g(a)^{2:(f(x) \approx g(x)\lor f(x) \not\approx g(x))\cdotp \sigma}]$
\end{center}
where $\sigma := \{x\mapsto a\}$. Decisions and propagations  are always ground instances of literals from the first-order clauses, and are annotated with a level
and a justification clause, in case of a decision a tautology. Now with respect to $\Gamma$ clause $C_3$ is false with grounding $\sigma$, and rule Conflict is applicable; see Section~\ref{subsec:scleqrules} for details
on the inference rules.
In general, clauses and justifications are considered variable disjoint, but for simplicity of the presentation of this example,
we repeat variable names here as long as the same ground substitution is shared. The maximal literal in $C_3\sigma$ is $(f(x) \not\approx h(x))\sigma$ and a
rewrite refutation using the ground equations from the trail results in the justification clause
\begin{center}
$(g(x) \not\approx g(x)\lor f(x) \not\approx g(x)\lor f(x) \not\approx g(x) \lor h(x) \not\approx g(x))\cdotp \sigma$
\end{center}
where for the refutation justification clauses and all otherwise inferred clauses we use the grounding $\sigma$ for guidance, but operate on the clauses with variables.
The respective ground clause is smaller than $(f(x) \not\approx h(x))\sigma$, false with respect to $\Gamma$ and becomes our new conflict clause by an application of our inference rule Explore-Refutation.
It is simplified by our inference rules Equality-Resolution and Factorize, resulting in the finally learned clause
\begin{center}
$C_4 :=  h(x) \not\approx g(x) \lor f(x) \not\approx g(x)$
\end{center}
which is then used to apply rule Backtrack to the trail. Further details
on this example are available from the Appendix, Example~\ref{exa:impcondec}.
Observe that $C_4$ is strictly stronger than $C'_4$ the clause inferred by superposition and that $C_4$ cannot be inferred by superposition.
Thus \SCLEQ\ can infer stronger clauses than superposition for this example.

\paragraph{Related Work:} SCL(EQ) is based on ideas of SCL~\cite{AlagiWeidenbach15,FioriWeidenbach19,BrombergerFW21} but for the first
time includes a native treatment of first-order equality reasoning. Similar to~\cite{BrombergerFW21} propagations need not to be exhaustively applied,
the trail is built out of decisions and propagations of ground literals annotated by first-order clauses,  SCL(EQ) only learns non-redundant clauses,
but for the first time conflicts resulting out of a decision have to be considered, due to the nature of the equality relation.

There have been suggested several approaches to lift the idea of an inference guiding model assumption from propositional to
full first-order logic~\cite{FioriWeidenbach19,baumgartner2006lemma,bonacina2014sggs,Bonacina2015}. They do not provide
a native treatment of equality, e.g., via paramodulation or rewriting.

Baumgartner et al. describe multiple calculi that handle equality by using unit superposition style inference rules and are based on either
hyper tableaux~\cite{Baumgartner98} or DPLL~\cite{davis1962machine,davis1960computing}. Hyper tableaux fix a major problem of the well-known free variable tableaux,
namely the fact that free variables within the tableau are rigid, i.e., substitutions have to be applied to all occurrences of a free variable within the entire tableau.
Hyper tableaux with equality~\cite{baumgartner2007hyper} in turn integrates unit superposition
style inference rules into the hyper tableau calculus.

Another approach that is related to ours is the
model evolution calculus with equality ($\mathcal{ME_E}$) by Baumgartner et al.~\cite{baumgartner2005model,baumgartner2012model}
which lifts the DPLL calculus to first-order logic with equality. Similar to our approach, $\mathcal{ME_E}$ creates a candidate model
until a clause instance contradicts this model or all instances are satisfied by the model. The candidate model results from a
so-called context, which consists of a finite set of non-ground rewrite literals. Roughly speaking, a context literal specifies
the truth value of all its ground instances unless a more specific literal specifies the complement. Initially the model
satisfies the identity relation over the set of all ground terms. Literals within a context may be universal or parametric,
where universal literals guarantee all its ground instances to be true. If a clause contradicts the current model,
it is repaired by a non-deterministic split which adds a parametric literal to the current model. If the added literal
does not share any variables in the contradictory clause it is added as a universal literal.

Another approach by Baumgartner and Waldmann~\cite{baumgartner2009superposition} combined the superposition calculus with the Model Evolution calculus with equality.
In this calculus the atoms of the clauses are labeled as "split atoms" or "superposition atoms". The superposition part of the calculus then generates a model for the
superposition atoms while the model evolution part generates a model for the split atoms. Conversely, this means that if all atoms are labeled as "split atom",
the calculus behaves similar to the model evolution calculus. If all atoms are labeled as "superposition atom", it behaves like the superposition calculus.

Both the hyper tableaux calculus with equality and the model evolution calculus with equality allow only unit superposition applications,
while SCL(EQ) inferences are guided paramodulation inferences on clauses of arbitrary length.
The model evolution calculus with equality was revised and implemented in 2011~\cite{baumgartner2012model} and compares
its performance with that of hyper tableaux.
Model evolution performed significantly better, with more problems solved in all relevant TPTP~\cite{Sutcliffe17} categories,
than the implementation of the hyper tableaux calculus.

Plaisted et al.~\cite{plaisted2000ordered} present the Ordered Semantic Hyper-Linking (OSHL) calculus. OSHL is an instantiation based approach that repeatedly chooses ground instances of a non-ground input clause set such that the current model does not satisfy the current ground clause set. A further step repairs the current model such that it satisfies the ground clause set again. The algorithm terminates if the set of ground clauses contains the empty clause. OSHL supports rewriting and narrowing, but only with unit clauses. In order to handle non-unit clauses it makes use of other mechanisms such as Brand's Transformation~\cite{bachmair1998elimination}.

Inst-Gen~\cite{korovin2013inst} is an instantiation based calculus, that creates ground instances of the input first-order formulas which are forwarded to a SAT solver.
If a ground instance is unsatisfiable, then the first-order set is as well. If not then the calculus creates more instances. The Inst-Gen-EQ calculus~\cite{korovin2010iprover}
creates instances by extracting instantiations of unit superposition refutations of selected literals of the first-order clause set.
The ground abstraction is then
extended by the extracted clauses and an SMT solver then checks the satisfiability of the resulting set of equational and non-equational
ground literals.

In favor of a better structure we have moved all proofs to an Appendix.
%They will be made available in an extended version published as a research report, such as arXiv.
The rest of the paper is organized as follows. Section~\ref{preliminaries}
provides basic formalisms underlying \SCLEQ. The rules of the calculus are presented in Section~\ref{scleq_rules}. Soundness and completeness
results are provided in Section~\ref{soundness_completeness}. We end with a discussion of obtained results and future work, Section~\ref{conclusion}.
The main contribution of this paper is the \SCLEQ\ calculus that only learns non-redundant clauses, permits subset based redundancy elimination and rewriting,
and its soundness and completeness.

\section{Preliminaries} \label{preliminaries}

We assume a standard first-order language with equality and signature $\Sigma = (\opers, \emptyset)$ where the only predicate symbol is equality $\approx$. $N$ denotes a set of clauses,
$C,D$ denote clauses, $L,K,H$ denote equational literals, $A,B$ denote equational atoms, $t,s$ terms from $\term(\opers,\varset)$ for an infinite set of variables $\varset$,
$f,g,h$ function symbols from $\opers$, $a,b,c$ constants from $\opers$ and $x,y,z$ variables from $\varset$.
The function $comp$ denotes the complement of a literal. We write $s\not\approx t$ as a shortcut for $\neg(s \approx t)$.
The literal $s\,\#\, t$ may denote both $s\approx t$ and $s\not\approx t$. The semantics of first-order logic and semantic entailment $\models$ is
defined as usual.

By $\sigma,\tau,\delta$ we denote substitutions, which are total mappings from variables to terms. Let $\sigma$ be a substitution, then its finite domain is defined as $dom(\sigma) := \{x \mid x\sigma \not= x\}$ and its codomain is defined as $codom(\sigma) = \{t \mid x\sigma = t, x \in dom (\sigma)\}$. We extend their application to literals, clauses and sets of such objects in the usual way.
A term, literal, clause or sets of these objects is ground if it does not contain any variable.
A substitution $\sigma$ is $ground$ if $codom(\sigma)$ is ground. A substitution $\sigma$ is $grounding$ for a term $t$, literal $L$, clause $C$ if $t\sigma$, $L\sigma$, $C\sigma$ is ground,
respectively. By $C\cdotp \sigma$, $L\cdotp \sigma$ we denote a closure consisting of a clause $C$, literal $L$ and a grounding substitution $\sigma$, respectively.
The function $gnd$ computes the set of all ground
instances of a literal, clause, or clause set.
The function $\mgu$ denotes the most general unifier of terms, atoms, literals, respectively. We assume that mgus do not introduce fresh variables and that they are idempotent.

The set of positions $\mpos(L)$ of a literal (term $\mpos(t)$) is inductively defined as usual.
The notion $L|_p$ denotes the subterm of a literal $L$ ($t|_p$ for term $t$) at position $p\in \mpos(L)$ ($p\in \mpos(t)$).
The replacement of a subterm of a literal $L$ (term $t$) at position $p\in \mpos(L)$ ($p\in \mpos(t)$) by a term $s$ is denoted by $L[s]_p$ ($t[s]_p$).
For example, the term $f(a,g(x))$ has the positions $\{\epsilon, 1, 2, 21\}$, $f(a,g(x))|_{21} = x$ and $f(a,g(x))[b]_2$ denotes the term $f(a,b)$.

Let $R$ be a set of rewrite rules $l \rightarrow r$, called a \emph{term rewrite system} (TRS). The rewrite relation $\rightarrow_R \subseteq T(\opers, \varset)\times T(\opers, \varset)$ is defined as usual by
$s\rightarrow_R t$ if there exists $(l\rightarrow r)\in R$, $p\in \mpos(s)$, and a matcher $\sigma$,
such that $s|_p=l\sigma$ and $t = s[r\sigma]_p.$
We write $s=t{\downarrow_R}$ if s is the normal form of $t$ in the rewrite relation $\rightarrow_R$. We write $s\,\#\, t =(s'\,\#\, t'){\downarrow_R}$ if $s$ is the normal form of $s'$ and $t$ is the normal form of $t'$. A rewrite relation is terminating
if there is no infinite descending chain $t_0\rightarrow t_1\rightarrow ...$ and confluent if $t \orewrites{*} s \rightarrow^* t'$ implies $t \leftrightarrow^* t'$.
A rewrite relation is convergent if it is terminating and confluent. A rewrite order is a irreflexive and transitive rewrite relation.
A TRS $R$ is terminating, confluent, convergent, if the rewrite relation $\rightarrow_R$ is terminating, confluent, convergent, respectively.
A term $t$ is called irreducible by a TRS $R$ if no rule from $R$ rewrites $t$.
Otherwise it is called reducible. A literal, clause is irreducible if all of its terms are irreducible, and reducible otherwise.
A substitution $\sigma$ is called irreducible if any $t\in codom(\sigma)$ is irreducible, and reducible otherwise.

Let $\prec_T$ denote a well-founded rewrite ordering on terms which is total on ground terms and for all ground terms $t$ there exist only finitely many ground terms $s \prec_T t$.
We call $\prec_T$ a \emph{desired} term ordering.
We extend $\prec_T$ to equations by assigning the multiset $\{s,t\}$ to positive equations $s\approx t$ and
$\{s,s,t,t\}$ to inequations $s\not\approx t$.
Furthermore, we identify $\prec_T$ with its multiset extension comparing multisets of literals.
For a (multi)set of terms $\{t_1,\ldots,t_n\}$ and a term $t$, we define $\{t_1,\ldots,t_n\} \prec_T t$ if $\{t_1,\ldots,t_n\} \prec_T \{t\}$.
For a (multi)set of Literals $\{L_1,\ldots,L_n\}$ and a term $t$, we define $\{L_1,\ldots,L_n\} \prec_T t$ if $\{L_1,\ldots,L_n\} \prec_T \{\{t\}\}$.
Given a ground term $\beta$ then $\gnd_{\prec_T\beta}$
computes the set of all ground instances of a literal, clause, or clause set where the groundings are smaller than
$\beta$ according to the ordering $\prec_T$. Given a set (sequence) of ground literals $\Gamma$ let $\conv(\Gamma)$ be a convergent rewrite system out of the positive equations in $\Gamma$ using $\prec_T$.

Let $\prec$ be a well-founded, total, strict ordering on
ground literals, which is lifted to clauses and clause sets by its respective multiset extension. We
overload $\prec$ for literals, clauses, clause sets if the meaning is clear from the context.
The ordering is lifted to the non-ground case via instantiation: we define $C \prec D$
if for all grounding substitutions $\sigma$ it holds $C\sigma \prec D\sigma$.
Then we define $\preceq$ as the reflexive closure of $\prec$ and $N^{\preceq C} := \{D \mid D\in N \;\text{and}\; D\preceq C\}$
and use the standard superposition style notion of redundancy~\cite{bachmair1994rewrite}.

\begin{definition}[Clause Redundancy] \label{appr_4:prelim:def:redundancy}
  A ground clause $C$ is \emph{redundant} with respect to a set $N$
  of ground clauses and an ordering $\prec$ if $N^{\preceq C} \models C$.
  A clause $C$ is \emph{redundant} with respect to a clause set
  $N$ and an ordering $\prec$ if for all $C'\in gnd(C)$, $C'$ is redundant with respect to $gnd(N)$.
\end{definition}

\section{The \SCLEQ~Calculus}
\label{scleq_rules}

We start the introduction of the calculus by defining the ingredients of an \SCLEQ~state.

\begin{definition}[Trail]
  A \emph{trail} $\Gamma:=[L_1^{i_1:C_1\cdotp\sigma_1},...,L_n^{i_n:C_n\cdotp\sigma_n}]$ is a consistent sequence of ground equations and inequations
  where $L_j$ is annotated by a level $i_j$ with $i_{j-1}\leq i_j$, and a closure $C_j\cdotp\sigma_j$. We omit the annotations if they
  are not needed in a certain context.
  A ground literal $L$ is true in $\Gamma$ if $\Gamma \models L$.
  A ground literal $L$ is false in $\Gamma$ if $\Gamma \models comp(L)$.
  A ground literal $L$ is undefined in $\Gamma$ if $\Gamma\not\models L$ and $\Gamma\not\models comp(L)$. Otherwise it is defined.
  For each literal $L_j$ in $\Gamma$ it holds that $L_j$ is undefined in $[L_1,...,L_{j-1}]$ and irreducible by $\conv(\{L_1,...,L_{j-1}\})$.
\end{definition}
%A trail $\Gamma$ is inconsistent iff there is an inequation $s\not\approx t\in \Gamma$
%such that $s{\downarrow_{\conv(\Gamma)}} = t{\downarrow_{\conv(\Gamma)}}$ with respect to $\prec_T$.
The above definition of truth and undefinedness is extended to clauses in the obvious way. The notions of true, false, undefined can be parameterized by a
ground term $\beta$ by saying that $L$ is $\beta$-undefined in a trail $\Gamma$ if $\beta\prec_T L$ or $L$ is undefined. The notions of a $\beta$-true, $\beta$-false term are
restrictions of the above notions to literals smaller $\beta$, respectively. All \SCLEQ\ reasoning is layered with respect to a ground term $\beta$.

\begin{definition}
  Let $\Gamma$ be a trail and $L$ a ground literal such that $L$ is defined in $\Gamma$.
  By $\core(\Gamma;L)$ we denote a minimal subsequence $\Gamma' \subseteq \Gamma$ such that $L$ is defined in $\Gamma'$.
  By $\cores(\Gamma;L)$ we denote the set of all cores.
\end{definition}

Note that $core(\Gamma;L)$ is not necessarily unique. There can be multiple cores for a given trail $\Gamma$ and ground literal $L$.

\begin{definition}[Trail Ordering]
  Let $\Gamma:=[L_1,...,L_n]$ be a trail. The (partial) trail ordering $\prec_\Gamma$ is the sequence ordering given by $\Gamma$, i.e.,  $L_i\prec_\Gamma L_j$ if $i < j$ for all $1\leq i,j\leq n$.
\end{definition}

\begin{definition}[Defining Core and Defining Literal]
 For a trail $\Gamma$ and a sequence of literals $\Delta \subseteq \Gamma$  we write $max_{\prec_\Gamma}(\Delta)$ for the largest literal in $\Delta$ according to the trail ordering $\prec_\Gamma$.
 Let $\Gamma$ be a trail and $L$ a ground literal such that $L$ is defined in $\Gamma$.
 Let $\Delta\in cores(\Gamma;L)$ be a sequence of literals where $max_{\prec_\Gamma}(\Delta) \preceq_\Gamma max_{\prec_\Gamma}(\Lambda)$
 for all $\Lambda\in cores(\Gamma;L)$, then $\max_\Gamma(L) := \max_{\prec_\Gamma}(\Delta)$ is called the \emph{defining literal} and $\Delta$ is called a \emph{defining core} for $L$ in $\Gamma$. If $cores(\Gamma;L)$
 contains only the empty core, then $L$ has no \emph{defining literal} and no \emph{defining core}.
\end{definition}

Note that there can be multiple defining cores but only one defining literal for any defined literal $L$.
For example, consider a trail $\Gamma :=[f(a)\approx f(b)^{1:C_1\cdotp\sigma_1},a\approx b^{2:C_2\cdotp\sigma_2},b\approx c^{3:C_3\cdotp\sigma_3}]$ with an ordering $\prec_T$ that
orders the terms of the equations from left to right,
and a literal $g(f(a))\approx g(f(c))$. Then the defining cores are $\Delta_1 :=[a\approx b,b\approx c]$ and $\Delta_2 :=[f(a)\approx f(b),b\approx c]$.
The defining literal, however, is in both cases $b\approx c$. Defined literals that have no defining core and therefore no defining literal are literals that are trivially false or true.
Consider, for example, $g(f(a))\approx g(f(a))$. This literal is trivially true in $\Gamma$. Thus an empty subset of $\Gamma$ is sufficient to show that $g(f(a))\approx g(f(a))$ is defined in $\Gamma$.
%Literals in a trail $\Gamma$ are annotated with
%a \emph{level} and a \emph{closure}, i.e., they have the form $L^{k:C\cdotp\sigma}$.

\begin{definition}[Literal Level]
 Let $\Gamma$ be a trail. A ground literal $L \in \Gamma$ is of \emph{level} $i$ if $L$ is annotated with $i$ in $\Gamma$.
 A defined ground literal $L\not\in \Gamma$ is of level $i$ if the defining literal of $L$ is of level $i$. If $L$ has no defining literal, then $L$ is of level $0$. A ground clause $D$ is of level $i$ if $i$ is the maximum level of a literal in $D$.
\end{definition}

The restriction to minimal subsequences for the defining literal and
definition of a level eventually guarantee that learned clauses are smaller in
the trail ordering. This enables completeness in combination with learning
non-redundant clauses as shown later.

\begin{lemma}
  \label{lem:lvlfix}
  Let $\Gamma_1$ be a trail and $K$ a defined literal that is of level $i$ in $\Gamma_1$. Then $K$ is of level $i$ in a trail $\Gamma:=\Gamma_1,\Gamma_2$.
\end{lemma}

\begin{definition}
  Let $\Gamma$ be a trail and $L\in\Gamma$ a literal. $L$ is called a \emph{decision literal} if $\Gamma = \Gamma_0,K^{i:C\cdotp\tau}, L^{i+1:C'\cdotp\tau'}, \Gamma_1$. Otherwise $L$ is called a \emph{propagated literal}.
\end{definition}

In our above example $g(f(a))\approx g(f(c))$ is of level $3$ since the defining literal $b\approx c$ is annotated with $3$. $a\not\approx b$ on the other hand is of level $2$.

We define a well-founded total strict ordering which is induced by the trail and with which non-redundancy is proven in Section~\ref{soundness_completeness}. Unlike SCL~\cite{FioriWeidenbach19,BrombergerFW21} we
use this ordering for the inference rules as well. In previous SCL calculi, conflict resolution automatically chooses the greatest literal and resolves with this literal. In~\SCLEQ\ this is generalized.
Coming back to our running example above, suppose we have a conflict clause $f(b)\not\approx f(c)\lor b\not\approx c$.
The defining literal for both inequations is $b\approx c$. So we could do paramodulation inferences with both literals. The following ordering makes this non-deterministic choice deterministic.

\begin{definition}[Trail Induced Ordering] \label{appr_4:trail_ind_order}
  Let $\Gamma:=[L_1^{i_1:C_1\cdotp\sigma_1},...,L_n^{i_n:C_n\cdotp\sigma_n}]$ be a trail, $\beta$ a ground term such that $\{L_1,...,L_n\}\prec_T \beta$
  and $M_{i,j}$ all $\beta$-defined ground literals
  not contained in $\Gamma \,\cup\, comp(\Gamma)$: for a  defining literal $max_\Gamma(M_{i,j}) = L_i$ and for two literals $M_{i,j}$, $M_{i,k}$ we have $j<k$ if $M_{i,j} \prec_T M_{i,k}$.
  %Let $L_p$ be the smallest $\prec_\Gamma$ literal with level greater zero.
  The trail induces a total well-founded strict order $\prec_{\Gamma^*}$ on $\beta\mhyphen\mathit{defined}$ ground literals $M_{k,l}, M_{m,n}$, $L_i$, $L_j$ of level greater than zero, where
  \begin{itemize}
  \item[1.] $M_{i,j}\prec_{\Gamma^*} M_{k,l}$ if $i<k$ or ($i=k$ and $j<l$)
  \item[2.]  $L_i\prec_{\Gamma^*} L_j$ if $L_i\prec_\Gamma L_j$
  \item[3.]  $comp(L_i)\prec_{\Gamma^*} L_j$ if $L_i\prec_\Gamma L_j$
  \item[4.]  $L_i\prec_{\Gamma^*} comp(L_j)$ if $L_i\prec_\Gamma L_j$ or $i = j$
  \item[5.]  $comp(L_i)\prec_{\Gamma^*} comp(L_j)$ if $L_i\prec_\Gamma L_j$
  \item[6.]  $L_i\prec_{\Gamma^*} M_{k,l}$,  $\comp(L_i)\prec_{\Gamma^*} M_{k,l}$ if $i\leq k$
  \item[7.]  $M_{k,l}\prec_{\Gamma^*} L_i$,  $M_{k,l}\prec_{\Gamma^*} \comp(L_i)$ if $k<i$
  \end{itemize}
  and for all $\beta\mhyphen\mathit{defined}$ literals $L$ of level zero:
  \begin{itemize}
  \item[8.]  $\prec_{\Gamma^*} := \prec_T$
  \item[9.]  $L\prec_{\Gamma^*} K$ if $K$ is of level greater than zero and $K$ is $\beta\mhyphen\mathit{defined}$
  \end{itemize}
  and can eventually be extended to $\beta\mhyphen\mathit{undefined}$ ground literals $K, H$ by
  \begin{itemize}
  \item[10.]  $K \prec_{\Gamma^*} H$ if $K\prec_T H$
  \item[11.]  $L \prec_{\Gamma^*} H$ if $L$ is $\beta\mhyphen\mathit{defined}$
  \end{itemize}
  The literal ordering $\prec_{\Gamma^*}$ is extended to ground clauses by multiset extension and identified with $\prec_{\Gamma^*}$ as well.
\end{definition}

\begin{lemma}[Properties of $\prec_{\Gamma^*}$] \label{gammastar:properties}
  \begin{enumerate}
    \item  $\prec_{\Gamma^*}$ is well-defined. \label{gammastar:properties:welldefined}
    \item  $\prec_{\Gamma^*}$ is a total strict order, i.e. $\prec_{\Gamma^*}$ is irreflexive, transitive and total. \label{gammastar:properties:strictorder}
    \item $\prec_{\Gamma^*}$ is a well-founded ordering. \label{gammastar:properties:wellfounded}
    \end{enumerate}
\end{lemma}

\begin{example}
  Assume a trail $\Gamma := [a\approx b^{1:C_0\cdotp\sigma_0}, c\approx d^{1:C_1\cdotp\sigma_1}, f(a')\not\approx f(b')^{1:C_2\cdotp\sigma_2}]$, select KBO as the term ordering $\prec_T$ where all symbols have weight one
  and $a\prec a'\prec b\prec b'\prec c\prec d\prec f$ and a ground term $\beta := f(f(a))$. According to the trail induced ordering we have that
  $a\approx b\prec_{\Gamma^*} c\approx d\prec_{\Gamma^*} f(a')\not\approx f(b')$
  by \ref{appr_4:trail_ind_order}.2. Furthermore we have that
  $$a\approx b\prec_{\Gamma^*} a\not\approx b\prec_{\Gamma^*} c\approx d\prec_{\Gamma^*} c\not\approx d\prec_{\Gamma^*} f(a')\not\approx f(b')\prec_{\Gamma^*} f(a')\approx f(b')$$
  by \ref{appr_4:trail_ind_order}.3 and \ref{appr_4:trail_ind_order}.4.
  Now for any literal $L$ that is $\beta\mhyphen\mathit{defined}$ in $\Gamma$ and the defining literal is $a\approx b$ it holds that $a\not\approx b\prec_{\Gamma^*} L\prec_{\Gamma^*} c\approx d$ by \ref{appr_4:trail_ind_order}.6 and \ref{appr_4:trail_ind_order}.7.
  This holds analogously for all literals that are $\beta\mhyphen\mathit{defined}$ in $\Gamma$ and the defining literal is $c\approx d$ or $f(a')\not\approx f(b')$. Thus we get:
  \begin{center}
    $\begin{array}{c}
       L_1 \prec_{\Gamma^*} ...  \prec_{\Gamma^*}  a\approx b\prec_{\Gamma^*} a\not\approx b\prec_{\Gamma^*} f(a)\approx f(b)\prec_{\Gamma^*} f(a)\not\approx f(b)\prec_{\Gamma^*} \\
       c\approx d\prec_{\Gamma^*} c\not\approx d\prec_{\Gamma^*} f(c)\approx f(d)\prec_{\Gamma^*} f(c)\not\approx f(d)\prec_{\Gamma^*}\\
       f(a')\not\approx f(b')\prec_{\Gamma^*} f(a')\approx f(b')\prec_{\Gamma^*} a'\approx b'\prec_{\Gamma^*} a'\not\approx b'\prec_{\Gamma^*} K_1\prec_{\Gamma^*} \ldots
       \end{array}$
  \end{center}
   where $K_i$ are the $\beta\mhyphen\mathit{undefined}$ literals and $L_j$ are the trivially defined literals.
\end{example}

\begin{definition}[Rewrite Step]
  A \emph{rewrite step} is a five-tuple $(s\# t\cdotp\sigma, s\# t\lor C\cdotp\sigma, R, S, p)$ and inductively defined as follows.
  The tuple $(s\# t\cdotp\sigma, s\# t\lor C\cdotp\sigma, \epsilon, \epsilon, \epsilon)$ is a rewrite step. Given rewrite steps $R, S$ and
  a position $p$ then $(s\# t\cdotp\sigma, s\# t\lor C\cdotp\sigma, R, S, p)$ is a \emph{rewrite step}.
  The literal $s\# t$ is called the \emph{rewrite literal}. In case $R,S$ are not $\epsilon$, the rewrite literal of $R$ is an equation.
\end{definition}

Rewriting is one of the core features of our calculus. The following definition describes a rewrite inference between two clauses. Note that unlike the superposition calculus we
allow rewriting below variable level.

\begin{definition}[Rewrite Inference]
  Let $I_1 := (l_1\approx r_1\cdotp\sigma_1,l_1\approx r_1\lor C_1\cdotp\sigma_1, R_1,\allowbreak L_1, p_1)$ and
  $I_2 := (l_2\# r_2\cdotp\sigma_2,l_2\# r_2\lor C_2\cdotp\sigma_2, R_2, L_2, p_2)$ be two variable disjoint rewrite steps where $r_1\sigma_1\prec_T l_1\sigma_1$,
  $(l_2\# r_2)\sigma_2|_p = l_1\sigma_1$ for some position $p$. We distinguish two cases:
  \begin{itemize}
  \item[1.] if $p\in\mpos(l_2\# r_2)$ and $\mu := \mgu((l_2\# r_2)|_p,l_1)$ then $(((l_2\# r_2)[r_1]_{p})\mu\cdotp \sigma_1\sigma_2,$\\
  $\allowbreak ((l_2\# r_2)[r_1]_{p})\mu\lor C_1\mu \lor C_2\mu\cdotp \sigma_1\sigma_2,\allowbreak I_1,\allowbreak I_2,\allowbreak p) $ is the result of a rewrite inference.
  \item[2.] if $p\not\in\mpos(l_2\# r_2)$ then let $(l_2\# r_2)\delta$ be the most general instance of $l_2\# r_2$ such that $p\in\mpos((l_2\# r_2)\delta)$, $\delta$ introduces only fresh variables
    and $(l_2\# r_2)\delta\sigma_2\rho = (l_2\# r_2)\sigma_2$  for some minimal $\rho$.
    Let $\mu := \mgu((l_2\# r_2)\delta|_p,l_1)$. Then\\
    $((l_2\# r_2)\delta[r_1]_{p}\mu\cdotp \sigma_1\sigma_2\rho, (l_2\# r_2)\delta[r_1]_{p}\mu\lor C_1\mu \lor C_2\delta\mu\cdotp \sigma_1\sigma_2\rho, I_1, I_2,p) $ is the result of a rewrite inference.
  \end{itemize}
\end{definition}

\begin{lemma}
  \label{rewrite_inference_properties}
  Let $I_1 := (l_1\approx r_1\cdotp\sigma_1,l_1\approx r_1\lor C_1\cdotp\sigma_1, R_1, L_1, p_1)$ and
  $I_2 := (l_2\# r_2\cdotp\sigma_2,\allowbreak l_2\# r_2\lor C_2\cdotp\sigma_2, R_2, L_2, p_2)$ be two variable disjoint rewrite steps where $r_1\sigma_1\prec_T l_1\sigma_1$,
  $(l_2\# r_2)\sigma_2|_p = l_1\sigma_1$ for some position $p$. Let $I_3 := (l_3\# r_3\cdotp\sigma_3,l_3\# r_3\lor C_3\cdotp\sigma_3, I_1, I_2, p)$ be the result of a rewrite inference.
  Then:
  \begin{enumerate}
  \item $C_3\sigma_3 = (C_1\lor C_2)\sigma_1\sigma_2$ and $l_3\# r_3\sigma_3 = (l_2\# r_2)\sigma_2[r_1\sigma_1]_p$.
  \item $(l_3\# r_3)\sigma_3\prec_T (l_2\# r_2)\sigma_2$
  \item If $N\models( l_1\approx r_1\lor C_1) \land (l_2\# r_2\lor C_2)$ for some set of clauses $N$, then $N\models l_3\# r_3\lor C_3$
\end{enumerate}
\end{lemma}
%\begin{proof}
%Trivial, since equivalent to a superposition inference.
%\end{proof}

Now that we have defined rewrite inferences we can use them to define a \emph{reduction chain application} and a \emph{refutation}, which are sequences of rewrite steps.
Intuitively speaking, a \emph{reduction chain application} reduces a literal in a clause with literals in $conv(\Gamma)$ until it is irreducible.
A \emph{refutation} for a literal $L$ that is $\beta\mhyphen \mathit{false}$ in $\Gamma$ for a given $\beta$, is a sequence of rewrite steps with literals in $\Gamma,L$ such that $\bot$ is inferred.
Refutations for the literals of the conflict clause will be examined during conflict resolution by the rule Explore-Refutation.

\begin{definition}[Reduction Chain]
  Let $\Gamma$ be a trail.
  A \emph{reduction chain} $\mathcal{P}$ from $\Gamma$ is a sequence of rewrite steps $[I_1,...,I_m]$
  such that for each $I_i = (s_i\# t_i\cdotp\sigma_i, s_i\# t_i\lor C_i\cdotp\sigma_i, I_j, I_k, p_i)$ either
  \begin{enumerate}
  \item $s_i\# t_i^{n_i:s_i\# t_i\lor C_i\cdotp\sigma}$ is contained in $\Gamma$ and $I_j = I_k = p_i = \epsilon$ or
  \item $I_i$ is the result of a rewriting inference from rewrite steps $I_j, I_k$ out of  $[I_1,...,I_m]$ where $j, k < i$.
  \end{enumerate}
  Let $(l\,\#\, r)\delta^{o:l\,\#\, r \lor C \cdotp\delta}$ be an annotated ground literal.
  A \emph{reduction chain application} from $\Gamma$ to $l\,\#\, r$ is a reduction chain  $[I_1,...,I_m]$ from $\Gamma,(l\,\#\, r)\delta^{o:l\,\#\, r \lor C \cdotp\delta}$ such
  that  $l\delta{\downarrow_{\conv(\Gamma)}} = s_m\sigma_m$ and $r\delta{\downarrow_{\conv(\Gamma)}} = t_m\sigma_m$.
  We assume reduction chain applications to be minimal, i.e., if any rewrite step is removed from the sequence it is no longer a reduction chain application.
\end{definition}

\begin{definition}[Refutation]
  Let $\Gamma$ be a trail and $(l\,\#\, r)\delta^{o:l\,\#\, r \lor C \cdotp\delta}$ an annotated ground literal that is $\beta\mhyphen \mathit{false}$ in $\Gamma$ for a given $\beta$.
  A \emph{refutation} $\mathcal{P}$ from $\Gamma$ and $l\,\#\, r $ is a reduction chain $[I_1,...,I_m]$
  from $\Gamma,(l\,\#\, r)\delta^{o:l\,\#\, r \lor C \cdotp\delta}$ such that $(s_m\# t_m)\sigma_m = s \not\approx s$ for some $s$.
  We assume refutations to be minimal, i.e., if any rewrite step $I_k$, $k<m$ is removed from the refutation, it is no longer a refutation.
\end{definition}

\subsection{The \SCLEQ\ Inference Rules} \label{subsec:scleqrules}

We can now define the rules of our calculus based on the previous definitions.
A \emph{state} is a six-tuple $(\Gamma; N; U; \beta; k; D)$ similar to the SCL calculus,
where $\Gamma$ a sequence of annotated ground literals,
$N$ and $U$ the sets of initial and learned clauses,
$\beta$ is a ground term such that for all $L\in \Gamma$ it holds $L\prec_T \beta$,
$k$ is the decision level,
and $D$ a status that is $\top$, $\bot$ or a closure $C\,\cdotp\sigma$. Before we propagate or decide any literal,
we make sure that it is irreducible in the current trail. Together with the design of $\prec_{\Gamma^*}$ this
eventually enables rewriting as a simplification rule.

\bigskip
\longnamerules{Propagate}
{$(\Gamma; N; U; \beta; k;\top)$}
{$(\Gamma,s_m\# t_m\sigma_m^{{k:(s_m\# t_m\lor C_m)\cdotp\sigma_m}}; N; U; \beta; k;\top)$}
{ provided there is a $C \in (N \cup U)$,
  $\sigma$ grounding for $C$,
  $C = C_0 \lor C_1 \lor L$,
  $\Gamma \models \neg C_0\sigma$,
  $C_1\sigma = L\sigma \lor ... \lor L\sigma$,
  $C_1 = L_1 \lor ... \lor L_n$,
  $\mu = mgu(L_1,..., L_n, L)$
  $L\sigma$ is $\beta\mhyphen \mathit{undefined}$ in $\Gamma$,
  $(C_0 \lor L)\mu\sigma \prec_T \beta$,
  $\sigma$ is irreducible by $conv(\Gamma)$,
  $[I_1,\ldots,I_m]$ is a reduction chain application from $\Gamma$ to $L\sigma^{k:(L\lor C_0)\mu\cdotp\sigma}$ where
  $I_m = (s_m\# t_m\cdotp\sigma_m, s_m\# t_m\lor C_m\cdotp\sigma_m, I_j, I_k, p_m)$.
}{\SCLEQ}

\bigskip
Note that the definition of Propagate also includes the case where $L\sigma$ is irreducible by $\Gamma$.
In this case $L = s_m\# t_m$ and $m=1$. The rule Decide below, is similar to Propagate, except for the
subclause $C_0$ which must be $\beta\mhyphen \mathit{undefined}$ or $\beta\mhyphen \mathit{true}$ in $\Gamma$, i.e., Propagate cannot be applied and the decision literal is annotated by a tautology.

\bigskip
\longnamerules{Decide}
{$(\Gamma; N; U; \beta; k;\top)$}
{$(\Gamma,s_m\# t_m\sigma_m^{{k+1:(s_m\# t_m\lor comp(s_m\# t_m))\cdotp\sigma_m}};\allowbreak N;\allowbreak U;\allowbreak \beta;\allowbreak k+1;\top)$}
{ provided there is a $C \in (N \cup U)$,
  $\sigma$ grounding for $C$,
  $C = C_0 \lor L$,
  $C_0\sigma$ is $\beta\mhyphen \mathit{undefined}$ or $\beta\mhyphen \mathit{true}$ in $\Gamma$,
  %$C_1\sigma = L\sigma \lor ... \lor L\sigma$,
  $L\sigma$ is $\beta\mhyphen \mathit{undefined}$ in $\Gamma$,
  $(C_0 \lor L)\sigma \prec_T \beta$,
  $\sigma$ is irreducible by $conv(\Gamma)$,
  $[I_1,\ldots,I_m]$ is a reduction chain application from $\Gamma$ to $L\sigma^{k+1:L\lor C_0\cdotp\sigma}$ where
  $I_m = (s_m\# t_m\cdotp\sigma_m, s_m\# t_m\lor C_m\cdotp\sigma_m, I_j, I_k, p_m)$.
}{\SCLEQ}

\bigskip

\shortrules{Conflict}
{$(\Gamma; N; U; \beta; k;\top)$}
{$(\Gamma; N; U; \beta; k; D)$}
{provided there is a $D' \in (N \cup U)$,
  $\sigma$ grounding for $D'$,
  $D'\sigma$ is $\beta\mhyphen \mathit{false}$ in $\Gamma$,
  $\sigma$ is irreducible by $conv(\Gamma)$,
  $D =\bot$ if $D'\sigma$ is of level $0$ and $D = D'\cdotp\sigma$ otherwise.
}{\SCLEQ}{10}

\bigskip

For the non-equational case, when a conflict clause is found by an SCL calculus~\cite{FioriWeidenbach19,BrombergerFW21}, the complements of its
first-order ground literals are contained in the trail. For equational literals this is not the
case, in general. The proof showing $D$ to be $\beta\mhyphen\mathit{false}$ with respect to $\Gamma$ is a rewrite
proof with respect to $\conv(\Gamma)$.
This proof needs to be analyzed to eventually perform paramodulation steps
on $D$ or to replace $D$ by a $\prec_{\Gamma^*}$ smaller $\beta\mhyphen\mathit{false}$ clause showing up in the proof.

\bigskip
\shortrules{Skip}
{$(\Gamma,K^{l:C\cdotp\tau},L^{k:C'\cdotp\tau'}; N; U; \beta; k; D\,\cdotp\sigma)$}
{$(\Gamma,K^{l:C\cdotp\tau}; N; U; \beta; l; D\,\cdotp\sigma)$}
{if $D\sigma$ is $\beta\mhyphen \mathit{false}$ in $\Gamma,K^{l:C\cdotp\tau}$.
}{\SCLEQ}{4}

\bigskip

The Explore-Refutation rule is the FOL with Equality counterpart to the resolve rule in CDCL or SCL.
While in CDCL or SCL complementary literals of the conflict clause are present
on the trail and can directly be used for resolution steps, this needs a generalization
for FOL with Equality. Here, in general, we need to look at (rewriting) refutations of the conflict
clause and pick an appropriate clause from the refutation as the next conflict clause.

\bigskip

\longnamerules{Explore-Refutation}
{$(\Gamma,L; N; U; \beta; k;(D\lor s\,\#\,t)\cdotp\sigma))$}
{$(\Gamma,L; N; U; \beta; k;(s_j\# t_j\lor C_j)\cdotp\sigma_j$)}
{ if $(s\,\#\,t)\sigma$ is strictly $\prec_{\Gamma^*}$ maximal in $(D\lor s\,\#\,t)\sigma$,
  $L$ is the defining literal of $(s\,\#\,t)\sigma$,
  $[I_1,...,I_m]$ is a refutation from $\Gamma$ and $(s\,\#\,t)\sigma$,
  $I_j = (s_j\# t_j\cdotp\sigma_j, (s_j\# t_j\lor C_j)\cdotp\sigma_j, I_l, I_k, p_j)$, $1\leq j\leq m$,
  $ (s_j\,\#\,t_j\lor C_j)\sigma_j \prec_{\Gamma^*} (D\lor s\,\#\,t)\sigma$, $(s_j\# t_j\lor C_j)\sigma_j$ is $\beta\mhyphen\mathit{false}$ in $\Gamma$.
}{\SCLEQ}

\bigskip
\longnamerules{Factorize}
{$(\Gamma;N;U;\beta;k;(D\lor L \lor L')\cdot\sigma)$}
{$(\Gamma;N;U;\beta;k;(D\lor L)\mu\cdot\sigma)$}
{provided $L\sigma = L'\sigma$, and $\mu =\mgu(L,L')$.}
{\SCLEQ}

\bigskip
\longnamerules{Equality-Resolution}
{$(\Gamma; N; U;\beta; k; (D \lor s\not\approx s')\cdotp\sigma)$}
{$ (\Gamma; N; U;\beta; k; D\mu\,\cdotp\sigma)$}
{provided $s\sigma = s'\sigma$, $\mu =mgu(s,s')$.
}{\SCLEQ}

%\bigskip
%\longnamerules{Backtrack}
%{$(\Gamma, K,\Gamma';N;U; \beta;k;(D \lor L)\,\cdotp\sigma)$}
%{$(\Gamma;N;U\cup \{D \lor L\}; \beta; i;\top)$}
%{provided $K$ is a decision literal of level $i+1$,
% $L\sigma$ is of level $k$ and
% $D\sigma$ is of level $i$.
%}{\SCLEQ}

\bigskip
\shortrules{Backtrack}
{$(\Gamma,K,\Gamma';N;U;\beta;k;(D\lor L)\cdot\sigma)$}
{$(\Gamma;N;U\cup\{D\lor L\};\beta;j-i;\top)$}
{provided $D\sigma$ is of level $i'$ where $i'<k$, $K$ is of level $j$ and
  $\Gamma,K$ the minimal trail subsequence such that there is
  a grounding substitution $\tau$ with $(D\lor L)\tau$ $\beta$-false in $\Gamma,K$ but not in  $\Gamma$;
  $i=1$ if $K$ is a decision literal and $i=0$ otherwise.}{\SCLEQ}{14}

%  \bigskip
%  \longnamerules{Backtrack}
%  {$(\Gamma,K,\Gamma';N;U;\beta;k;(D\lor L)\cdot\sigma)$}
%  {$(\Gamma;N;U\cup\{D\lor L\};\beta;l;\top)$}
%  {provided $D\sigma$ is of level $i$ where $i<k$,
%   $L\sigma$ is of level $k$,
%   $K$ is of level $j$,
%   $l=j-1$ if $K$ is a decision literal and $l=j$ otherwise,
%   there is a grounding substitution $\tau$ with $(D\lor L)\tau$ $\beta\mhyphen\mathit{false}$ in $\Gamma,K$,
%   there is no grounding substitution $\delta$ with $(D\lor L)\delta$ $\beta\mhyphen\mathit{false}$ in $\Gamma$.}{\SCLEQ}

\bigskip
\shortrules{Grow}
{$(\Gamma;N;U; \beta;k;\top)$}
{$(\epsilon;N;U; \beta'; 0;\top)$}
{provided $\beta \prec_T \beta'$.
}{\SCLEQ}{6}

\bigskip
In addition to soundness and completeness of the \SCLEQ\ rules their tractability in practice is an important
property for a successful implementation. In particular, finding propagating literals or detecting a false clause
under some grounding. It turns out that these operations are NP-complete, similar to first-order subsumption which
has been shown to be tractable in practice.

\begin{lemma}
  \label{lem:np-complete}
  Assume that all ground terms $t$ with $t\prec_T \beta$ for any $\beta$ are polynomial in the size of $\beta$.
  Then testing Propagate (Conflict) is NP-Complete, i.e., the problem of checking for a given clause $C$ whether there exists a grounding
  substitution $\sigma$ such that $C\sigma$ propagates (is false) is NP-Complete.
\end{lemma}

\begin{example}[SCL(EQ) vs. Superposition: Saturation]
  Consider the following clauses:
  $$N := \{C_1 := c\approx d\lor D, C_2 := a\approx b\lor c\not\approx d, C_3 := f(a)\not\approx f(b)\lor g(c)\not\approx g(d)\}$$
  where again we assume a KBO with all symbols having weight one, precedence $d\prec c\prec b\prec a\prec g\prec f$ and $\beta := f(f(g(a)))$.
  Suppose that we first decide $c\approx d$ and then propagate $a\approx b$: $\Gamma = [c\approx d^{1:c\approx d\lor c\not\approx d}, a\approx b^{1:C_2}]$.
  Now we have a conflict with $C_3$. Explore-Refutation applied to the conflict clause $C_3$ results in a paramodulation inference
  between $C_3$ and $C_2$. Another application of Equality-Resolution gives us the new conflict clause
  $C_4 := c\not\approx d\lor g(c)\not\approx g(d)$. Now we can Skip the last literal on the trail,
  which gives us $\Gamma = [c\approx d^{1:c\approx d\lor c\not\approx d}]$. Another application of the Explore-Refutation rule to $C_4$ using the decision justification clause
  followed by Equality-Resolution and Factorize gives us $C_5 := c\not\approx d$. Thus with \SCLEQ\ the following clauses remain:
  \begin{center}
    $\begin{array}{c}
       C'_1 = D \qquad C_5 = c\not\approx d\\
        C_3 = f(a)\not\approx f(b)\lor g(c)\not\approx g(d)\\
     \end{array}$
  \end{center}
  where we derived $C'_1$ out of $C_1$ by subsumption resolution~\cite{Weidenbach01handbook} using $C_5$.
  Actually, subsumption resolution is compatible with the general redundancy notion of \SCLEQ, see Lemma~\ref{appr_4:non-red}.
  Now we consider the same example with superposition and the very same ordering ($N_i$ is the clause set of the previous step and $N_0$ the initial clause set $N$).
  \begin{center}
    $\begin{array}{ll}
    N_0 &\Rightarrow_{Sup(C_2,C_3)}N_1\cup\{C_4 := c\not\approx d\lor g(c)\not\approx g(d)\}\\
      &\Rightarrow_{Sup(C_1,C_4)}N_2\cup\{C_5 := c\not\approx d\lor D\} \Rightarrow_{Sup(C_1,C_5)}N_3\cup\{C_6 := D\}\\
     \end{array}$
  \end{center}
  Thus superposition ends up with the following clauses:
  \begin{center}
    $\begin{array}{ll}
    C_2 = a\approx b\lor c\not\approx d & C_3 = f(a)\not\approx f(b)\lor g(c)\not\approx g(d)\\
    C_4 = c\not\approx d\lor g(c)\not\approx g(d) \quad & C_6 = D\\
     \end{array}$
  \end{center}
  The superposition calculus generates more and larger clauses.
\end{example}

\begin{example}[SCL(EQ) vs. Superposition: Refutation]
  Suppose the following set of clauses: $N := \{C_1 := f(x)\not\approx a\lor f(x)\approx b, C_2 := f(f(y))\approx y, C_3 := a\not\approx b\}$
  where again we assume a KBO with all symbols having weight one, precedence $b\prec a\prec f$ and $\beta := f(f(f(a)))$. A long refutation by the
  superposition calculus results in the following ($N_i$ is the clause set of the previous step and $N_0$ the initial clause set $N$):
  \begin{center}
    $\begin{array}{lll}
    N_0 & \Rightarrow_{Sup(C_1,C_2)} & N_1\cup \{C_4 := y\not\approx a\lor f(f(y))\approx b\}\\
      & \Rightarrow_{Sup(C_1,C_4)} & N_2\cup \{C_5 := a\not\approx b\lor f(f(y))\approx b\lor y\not\approx a\}\\
      & \Rightarrow_{Sup(C_2,C_5)} & N_3\cup \{C_6 := a\not\approx b\lor b\approx y\lor y\not\approx a\}\\
      & \Rightarrow_{Sup(C_2,C_4)} & N_4\cup \{C_7 := y\approx b\lor y\not\approx a\}\\
      & \Rightarrow_{EqRes(C_7)}   & N_5\cup \{C_8 := a\approx b\}\Rightarrow_{Sup(C_3,C_8)}N_6\cup \{\bot\}\\
     \end{array}$
  \end{center}
  The shortest refutation by the superposition calculus is as follows:
  \begin{center}
    $\begin{array}{lll}
      N_0 & \Rightarrow_{Sup(C_1,C_2)} & N_1\cup \{C_4 := y\not\approx a\lor f(f(y))\approx b\}\\
        & \Rightarrow_{Sup(C_2,C_4)} & N_2\cup\{C_5 := y\approx b\lor y\not\approx a\}\\
        & \Rightarrow_{EqRes(C_5)}   & N_3\cup\{C_6 := a\approx b\} \Rightarrow_{Sup(C_3,C_6)}N_4\cup\{\bot\}\\
     \end{array}$
  \end{center}
  In \SCLEQ~on the other hand we would always first propagate $a\not\approx b,f(f(a))\approx a$ and $f(f(b))\approx b$. As soon as $a\not\approx b$ and $f(f(a))\approx a$ are propagated we have a conflict with $C_1\{x\rightarrow f(a)\}$. So suppose in the worst case we propagate:
  $$\Gamma := [a\not\approx b^{0:a\not\approx b}, f(f(b))\approx b^{0:(f(f(y))\approx y)\{y\rightarrow b\}}, f(f(a))\approx a^{0:(f(f(y))\approx y)\{y\rightarrow a\}}]$$
  Now we have a conflict with $C_1\{x\rightarrow f(a)\}$. Since there is no decision literal on the trail, $\mathit{Conflict}$ rule immediately returns $\bot$ and we are done.
\end{example}

\section{Soundness and Completeness}
\label{soundness_completeness}
In this section we show soundness and refutational completeness of~\SCLEQ\ under the assumption of a regular run.
We provide the definition of a regular run and show that for a regular run all learned clauses are non-redundant according to our trail induced ordering. We start with the definition of a sound state.

\begin{definition}
\label{appr_4:soundness}
A state $(\Gamma; N; U; \beta; k; D)$ is sound if the following conditions hold:
\begin{enumerate}
 \item \label{appr_4:sound_1}$\Gamma$ is a consistent sequence of annotated literals,
 \item \label{appr_4:sound_2}for each decomposition $\Gamma = \Gamma_1, L\sigma^{i:(C \lor L)\cdotp \sigma}, \Gamma_2$ where $L\sigma$ is a propagated literal, we have that $C\sigma$ is  $\beta\mhyphen\mathit{false}$ in $\Gamma_1$, $L\sigma$ is $\beta\mhyphen\mathit{undefined}$ in $\Gamma_1$ and irreducible by $conv(\Gamma_1)$, $N\cup U \models (C \lor L)$ and $(C \lor L) \sigma\prec_T \beta$,
 \item \label{appr_4:sound_3}for each decomposition $\Gamma = \Gamma_1, L\sigma^{i:(L\lor comp(L))\cdotp\sigma}, \Gamma_2$ where $L\sigma$ is a decision literal, we have that $L\sigma$ is $\beta\mhyphen\mathit{undefined}$ in $\Gamma_1$ and irreducible by $conv(\Gamma_1)$, $N\cup U \models (L\lor comp(L))$ and $(L \lor comp(L))\sigma\prec_T \beta$,
 \item \label{appr_4:sound_4}$N \models U$,
 \item \label{appr_4:sound_5}if $D = C\clsr \sigma$, then $C\sigma$ is $\beta\mhyphen\mathit{false}$ in $\Gamma$, $N \cup U\models C$,
\end{enumerate}

\end{definition}
\begin{lemma}
\label{appr_4:init_state_sound}
 The initial state $(\epsilon; N; \emptyset; \beta; 0; \top)$ is sound.
\end{lemma}

\begin{definition}
 A run is a sequence of applications of \SCLEQ ~rules starting from the initial state.
\end{definition}

\begin{theorem}
  \label{appr_4:scleq_sound}
 Assume a state $(\Gamma;N;U;\beta;k;D)$ resulting from a run. Then $(\Gamma;N;U;\beta;k;D)$ is sound.
\end{theorem}

Next, we give the definition of a regular run. Intuitively speaking, in a regular run we are always allowed to do decisions except if
\begin{enumerate}
\item a literal can be propagated before the first decision and
\item the negation of a literal can be propagated.
\end{enumerate}
To ensure non-redundant learning we enforce at least one application of Skip during conflict resolution except for the special case of a conflict after a decision.

\begin{definition}[Regular Run] \label{appr_4:def:regular-run}
  A run is called \emph{regular} if
  \begin{enumerate}
  \item the rules $\mathit{Conflict}$ and $\mathit{Factorize}$ have precedence over all other rules,
  \item If $k=0$ in a state $(\Gamma; N;U;\beta;k;D)$, then $\mathit{Propagate}$ has precedence over $\mathit{Decide}$,
  \item If an annotated literal $L^{k:C\cdotp\sigma}$ could be added by an application of Propagate on $\Gamma$ in a state $(\Gamma; N;U;\beta;k;D)$ and $C\in N\cup U$,
    then the annotated literal $comp(L)^{k+1:C'\cdotp \sigma'}$ is not added by Decide on $\Gamma$,
  \item during conflict resolution $\mathit{Skip}$ is applied at least once, except if $\mathit{Conflict}$ is applied immediately after an application of $\mathit{Decide}$.
  \item if $\mathit{Conflict}$ is applied immediately after an application of $\mathit{Decide}$, then Backtrack is only applied in a state $(\Gamma,L'; N;U;\beta;k;D\cdotp\sigma)$
    if $L\sigma = comp(L')$ for some $L\in D$.
  \end{enumerate}
\end{definition}
%\begin{proof}
%  (Idea) By structural induction on an \SCLEQ~derivation and case analysis on the different rules of the calculus.\qed
%\end{proof}

%\begin{lemma}
%  \label{lit_lvl_stays}
%  Assume a state $(\Gamma;N;U;\beta;k; D)$ resulting from a regular run. If $\Gamma = \Gamma_1, K^{i+1}, \Gamma_2$ and $L$ is of level $i$ in $\Gamma$ and
% $$(\Gamma; N; U; \beta; k; D)\Rightarrow^*_{\SCLEQ} (\Gamma_1, K^{i+1},\Gamma_3; N; U'; \beta; k'; D') $$
% then $L$ is of level $i$ in $\Gamma_1, K^{i+1},\Gamma_3$.
%\end{lemma}

Now we show that any learned clause in a regular run is non-redundant according to our trail induced ordering.

\begin{lemma}[Non-Redundant Clause Learning] \label{appr_4:non-red}
Let $N$ be a clause set. The clauses learned during a regular run in \SCLEQ\ are not redundant with respect to $\prec_{\Gamma^*}$ and $N \cup U$.
For the trail only
non-redundant clauses need to be considered.
\end{lemma}

The proof of Lemma~\ref{appr_4:non-red} is based on the fact that conflict resolution eventually produces a clause smaller then
the original conflict clause with respect to $\prec_{\Gamma^*}$. All simplifications, e.g., contextual rewriting,
as defined in~\cite{bachmair1994rewrite,Weidenbach01handbook,WeidenbachWischnewskiCADE08,WeidenbachWischnewski10,Wischnewski12,GleissEtAl20},
are therefore compatible with Lemma~\ref{appr_4:non-red} and may be applied to the newly learned clause as long as they
respect the induced trail ordering. In detail, let $\Gamma$ be the trail before the application of rule Backtrack.
The  newly learned clause can be simplified according to the induced trail ordering $\prec_{\Gamma^*}$ as long
as the simplified clause is smaller with respect to $\prec_{\Gamma^*}$.

Another important consequence of Lemma~\ref{appr_4:non-red} is that newly learned clauses need not
to be considered for redundancy. Furthermore,  the \SCLEQ\ calculus always terminates, Lemma~\ref{lem:soundcomp:termination},
because there only finitely many non-redundant clauses with respect to a fixed $\beta$.

For dynamic redundancy, we have to consider the fact that the induced trail ordering changes. At this level, only redundancy criteria and simplifications that
are compatible with \emph{all} induced trail orderings may be applied. Due to the construction of the induced trail ordering, it is compatible with $\prec_T$ for
unit clauses.

\begin{lemma}[Unit Rewriting] \label{lem:soundcomplete:unitrewriting}
  Assume a state $(\Gamma;N;U;\beta;k;D)$ resulting from a regular run where the current level $k>0$ and a unit clause $l\approx r\in N$. Now assume a clause $C\lor L[l']_p \in N$ such that $l' = l\mu$ for some matcher $\mu$.
  Now assume some arbitrary grounding substitutions $\sigma'$ for $C\lor L[l']_p$, $\sigma$ for $l\approx r$ such that $l\sigma=l'\sigma'$ and $r\sigma\prec_T l\sigma$.
  Then $ (C\lor L[r\mu\sigma\sigma']_p)\sigma'\prec_{\Gamma^*} (C\lor L[l']_p)\sigma'$.
\end{lemma}

In addition, any notion that is based on a literal subset relationship is also compatible with ordering changes. The standard example is subsumption.

\begin{lemma}
  \label{lem:redundancy}
  Let $C,D$ be two clauses. If there exists a substitution $\sigma$ such that $C\sigma\subset D$, then $D$ is redundant with respect to $C$ and any $\prec_{\Gamma^*}$.
\end{lemma}

The notion of redundancy, Definition~\ref{appr_4:prelim:def:redundancy},
only supports a strict subset relation for Lemma~\ref{lem:redundancy}, similar to the superposition calculus. However, the newly generated clauses
of \SCLEQ\ are the result of paramodulation inferences~\cite{RobinsonWos69}. In a recent contribution to dynamic, abstract redundancy \cite{WaldmannEtAl20} it is shown that also the non-strict
subset relation in Lemma~\ref{lem:redundancy}, i.e., $C\sigma\subseteq D$, preserves completeness.

If all stuck states, see below Definition~\ref{appr_4:def:stuck-state}, with respect to a fixed $\beta$ are visited before increasing $\beta$ then
this provides a simple dynamic fairness strategy.

When unit reduction or any other form of supported rewriting is applied to clauses smaller than the current $\beta$, it can be applied independently from the current trail.
If, however, unit reduction is applied to clauses larger than the current $\beta$ then the calculus must do a restart to its initial state, in particular the trail must be emptied,
as for otherwise rewriting may result generating a conflict that did not exist with respect to the current trail before the rewriting.
This is analogous to a restart in CDCL once a propositional unit clause is derived and used for simplification.
More formally, we add the following new Restart rule to the calculus to reset the trail to its initial state after a unit reduction.

\bigskip
\shortrules{Restart}
{$(\Gamma;N;U; \beta;k;\top)$}
{$(\epsilon;N;U; \beta; 0;\top)$}
{}{\SCLEQ}{8}

\bigskip

Next we show refutation completeness of~\SCLEQ. To achieve this we first give a definition of a stuck state.
Then we show that stuck states only occur if all ground literals $L\prec_T\beta$ are $\beta\mhyphen\mathit{defined}$ in $\Gamma$ and not during conflict resolution.
Finally we show that conflict resolution will always result in an application of Backtrack. This allows us to show termination (without application of Grow) and refutational completeness.

\begin{definition}[Stuck State] \label{appr_4:def:stuck-state}
  A state $(\Gamma;N;U;\beta;k;D)$ is called \emph{stuck} if $D\neq \bot$ and none of the rules of the calculus, except for Grow, is applicable.
\end{definition}

\begin{lemma}[Form of Stuck States]
  \label{form_of_stuck_states}
  If a regular run (without rule Grow) ends in a stuck state $(\Gamma;N;U;\beta;k;D)$, then $D=\top$ and all ground literals $L\sigma\prec_T \beta$, where $L\lor C\in (N\cup U)$
  are $\beta\mhyphen\mathit{defined}$ in $\Gamma$.
\end{lemma}

\begin{lemma}
  \label{backtrack_finally_applicable}
  Suppose a sound state $(\Gamma; N;U;\beta;k; D)$ resulting from a regular run where $D\not\in\{\top,\bot\}$.
  If $Backtrack$ is not applicable then any set of applications of $\mathit{Explore\mhyphen Refutation}$, $Skip$, $\mathit{Factorize}$, $\mathit{Equality\mhyphen Resolution}$
  will finally result in a sound state $(\Gamma'; N;U;\beta;k;D')$, where $D' \prec_{\Gamma^*} D$. Then Backtrack will be finally applicable.
\end{lemma}

\begin{corollary}[Satisfiable Clause Sets] \label{coro:soundcomp:satclausesets}
  Let $N$ be a satisfiable clause set. Then any regular run without rule Grow will end in a stuck state, for any $\beta$.
\end{corollary}

Thus a stuck state can be seen as an indication for a satisfiable clause set. Of course, it remains to be investigated whether the clause set
is actually satisfiable. Superposition is one of the strongest approaches to detect satisfiability and constitutes a decision procedure
for many decidable first-order fragments~\cite{BachmairGanzingerEtAl93a,GanzingerNivelle99}.
Now given a stuck state and some specific ordering such as KBO, LPO, or some polynomial ordering~\cite{DershowitzPlaisted01handbook},
it is decidable whether the ordering can be instantiated from a stuck state such that $\Gamma$ coincides with the superposition model operator
on the ground terms smaller than $\beta$. In this case it can be effectively checked whether the clauses derived so far are actually saturated
by the superposition calculus with respect to this specific ordering. In this sense, \SCLEQ\ has the same power to decide satisfiability of
first-order clause sets than superposition.

\begin{definition}
 A regular run terminates in a state $(\Gamma; N;U;\beta;k; D)$ if $D=\top$ and no rule is applicable, or $D = \bot$.
\end{definition}

\begin{lemma} \label{lem:soundcomp:termination}
 Let $N$ be a set of clauses and $\beta$ be a ground term. Then any regular run that never uses Grow terminates.
\end{lemma}

\begin{lemma}
  \label{lem:soundcomp:unsat}
 If a regular run reaches the state $(\Gamma;N;U;\beta;k;\bot)$ then $N$ is unsatisfiable.
\end{lemma}

\begin{theorem}[Refutational Completeness]
  \label{theo:soundcomp:refutcomp}
  Let $N$ be an unsatisfiable clause set, and $\prec_T$ a desired term ordering.
  For any ground term $\beta$ where $gnd_{\prec_T\beta}(N)$ is unsatisfiable, any regular $\SCLEQ$ run
  without rule Grow will terminate by deriving $\bot$.
\end{theorem}

\section{Discussion} \label{conclusion}

We presented SCL(EQ), a new sound and complete calculus for reasoning in first-order logic with equality.
We will now discuss some of its aspects and present ideas for future work beyond the scope of this paper.

The trail induced ordering, Definition~\ref{appr_4:trail_ind_order}, is the result of letting the calculus
follow the logical structure of the clause set on the literal level and at the same time  supporting rewriting at the term
level. It can already be seen by examples on ground clauses over (in)equations over constants that this combination requires
a layered approach as suggested by Definition~\ref{appr_4:trail_ind_order}, see Example~\ref{exa:propsmeq} from the Appendix.

In case the calculus runs into a stuck state, i.e., the current trail is a model for the set of considered
ground instances, then the trail information can be effectively used for a guided continuation. For example,
in order to use the trail to certify a model, the trail literals can be used to guide the design of a lifted ordering
for the clauses with variables such that propagated trail literals are maximal in respective clauses.
Then it could be checked by superposition, if the current clause is saturated by such an ordering.
If this is not the case, then there must be a superposition inference larger than the
current $\beta$, thus giving a hint on how to extend $\beta$.
Another possibility is to try to extend the finite set of ground terms considered in a stuck state
to the infinite set of all ground terms by building extended equivalence classes following patterns
that ensure decidability of clause testing, similar to the ideas in~\cite{BrombergerFW21}.
If this fails, then again this information can be used to find an appropriate extension term $\beta$ for rule Grow.

In contrast to superposition, \SCLEQ\ does also inferences below variable level. Inferences in \SCLEQ\ are guided by
a false clause with respect to a partial model assumption represented by the trail. Due to this guidance and the
different style of reasoning this does not result in an explosion
in the number of possibly inferred clauses but also rather in the derivation of more general clauses, see Example~\ref{exa:rewvarlev} from the Appendix.

Currently, the reasoning with solely positive equations is done on and with respect to the trail.
It is well-known that also inferences from this type of reasoning can be used to speed up the overall
reasoning process. The SCL(EQ) calculus already provides all information for such a type of reasoning,
because it computes the justification clauses for trail reasoning via rewriting inferences. By an assessment of the quality of these clauses,
e.g., their reduction potential with respect to trail literals, they could also be added, independently from resolving a conflict.

The trail reasoning is currently defined with respect to rewriting. It could
also be performed by congruence closure~\cite{NelsonOppen80}.

Towards an implementation, the aspect of how to find
interesting ground decision or propagation literals for the trail can be treated similar to CDCL~\cite{MSS96,BayardoSchrag96,MoskewiczMadiganZhaoZhangMalik01,BiereEtAl09handbook}. A simple heuristic
may be used from the start, like counting the number of instance relationships of some ground literal
with respect to the clause set, but later on a bonus system can focus the search towards the structure
of the clause sets. Ground literals involved in a conflict or the process of learning a new clause
get a bonus or preference. The regular strategy requires the propagation of all ground unit clauses smaller
than $\beta$. For an implementation a propagation of the (explicit and implicit) unit clauses with variables to the trail will be
a better choice. This complicates the implementation of refutation proofs and rewriting (congruence closure),
but because every reasoning is layered by a ground term $\beta$ this can still be efficiently done.

\paragraph{Acknowledgments:}{%
This work was partly funded by DFG grant 389792660 as part of
TRR~248, see \url{https://perspicuous-computing.science}.
We thank the anonymous reviewers and Martin Desharnais for their thorough reading,
detailed comments, and corrections.
}


%

\newpage

\section*{Appendix}
\subsection{Proofs and Auxiliary Lemmas}

\subsubsection{Proof of Lemma~\ref{lem:lvlfix}}
Let $\Gamma_1$ be a trail and $K$ a defined literal that is of level $i$ in $\Gamma_1$. Then $K$ is of level $i$ in a trail $\Gamma:=\Gamma_1,\Gamma_2$.
\begin{proof}
  Assume a trail $\Gamma_1$ and a literal $K$ that is of level $i$ in $\Gamma_1$. Let $\Gamma:=\Gamma_1,\Gamma_2$ be a trail. Then we have two cases:
  \begin{enumerate}
    \item $K$ has no defining literal in $\Gamma_1$. Then $cores(\Gamma_1;K) = \{[]\}$ contains only the empty core and $K$ is of level $0$ in $\Gamma_1$. Then $cores(\Gamma;K) = \{[]\}$ as well and thus $K$ is of level $0$ in $\Gamma$.
    \item $K$ has a defining literal $L:=max_{\Gamma_1}(K)$ and $L$ is of level $i$. Then there exists a core $\Delta\in cores(\Gamma_1;K)$ such that $L$ is the maximum literal in $\Delta$ according to $\prec_{\Gamma}$ and for all $\Lambda\in cores(\Gamma_1;K)$ it holds $max_{\prec_{\Gamma}}(\Delta)\preceq_{\Gamma}max_{\prec_{\Gamma}}(\Lambda)$.
    Thus $\Delta$ is a defining core. Now any $\Lambda\in (cores(\Gamma;K)\setminus cores(\Gamma_1;K))$ has a higher maximum literal according to $\prec_\Gamma$. Thus $\Delta$ is also a defining core in $\Gamma$ and $L$ is the defining literal of $K$ in $\Gamma$ and thus $K$ is of level $i$ in $\Gamma$.
  \end{enumerate}
\end{proof}

\subsubsection{Auxiliary Lemmas for the Proofs of Lemma~\ref{gammastar:properties}}
\begin{lemma}
  \label{gammastar:properties:proofs:uniqueness}
Let $\Gamma$ be a trail. Then any literal in $\Gamma$ occurs exactly once.
\end{lemma}
\begin{proof}
Let  $\Gamma :=[L_1^{i_1:C_1\cdotp \sigma_1},...,L_n^{i_n:C_n\cdotp \sigma_n}]$. Now suppose there exist $L_i, L_j$ with $i<j$ and $1\leq i,j\leq n$ such that $L_i=L_j$. By definition of $\Gamma$, $L_j$ is undefined in $[L_1,...,L_i,...,L_{j-1}]$. But obviously $L_j$ is defined in $\Gamma$. Contradiction.
\end{proof}

\begin{lemma}
  \label{gammastar:properties:proofs:litlvlunique}
Let $\Gamma$ be a trail. If a literal $L$ is of level $i$, then it is not of level $j\not= i$.
\end{lemma}
\begin{proof}
Let $\Gamma$ be a trail. By lemma~\ref{gammastar:properties:proofs:uniqueness} any literal in $\Gamma$ is unique.
Suppose there exists a literal $L$ such that $L$ is of level $i$ and of level $j$. If the core is empty for $L$ then $L$ is of level $0$ by definition.
Otherwise there must exist cores $\Delta,\Lambda\in cores(\Gamma;L)$ such that $max_{\prec_\Gamma}(\Delta)\preceq_{\Gamma} max_{\prec_\Gamma}(\Lambda')$ and $max_{\prec_\Gamma}(\Lambda)\preceq_{\Gamma} max_{\prec_\Gamma}(\Lambda')$ for all $\Lambda'\in cores(\Gamma;L)$. But then $max_{\prec_\Gamma}(\Lambda)=max_{\prec_\Gamma}(\Delta)$. Contradiction.
\end{proof}

\begin{lemma}
  \label{gammastar:properties:proofs:alldefhavelevel}
  Let $L$ be a ground literal and $\Gamma$ a trail. If $L$ is defined in $\Gamma$ then $L$ has a level.
\end{lemma}
\begin{proof}
  Let $\Gamma$ be a trail. Suppose that $L$ is defined in $\Gamma$. Then it either has a defining literal or it has no defining literal. If it has a defining literal $K$, then the level of $K$ is the level of $L$. Since $K\in\Gamma$ it is annotated by a level. Thus $L$ has a level. If $L$ does not have a defining literal, then $L$ is of level $0$ by definition of a literal level.
\end{proof}

\subsubsection{Proof of Lemma~\ref{gammastar:properties}-\ref{gammastar:properties:welldefined}}
$\prec_{\Gamma^*}$ is well-defined.
\begin{proof}
  \label{well-defined-proof}
Suppose a trail $\Gamma :=[L_1^{i_1:C_1\cdotp \sigma_1},...,L_n^{i_n:C_n\cdotp \sigma_n}]$ and a term $\beta$ such that $\{L_1,...,L_n\}\prec_T\beta$.
We have to show that the rules \ref{appr_4:trail_ind_order}.1-\ref{appr_4:trail_ind_order}.11 are pairwise disjunct.
Consider the rules \ref{appr_4:trail_ind_order}.1-\ref{appr_4:trail_ind_order}.7. These rules are pairwise disjunct, if the sets $\{L_1,...,L_n\}$, $\{comp(L_1),...,comp(L_n)\}$ and $\{M_{i,j}~|~i\leq n\}$ are pairwise disjunct.
Obviously, $\{L_1,...,L_n\} \cap \{comp(L_1),...,comp(L_n)\} = \emptyset$.
Furthermore $(\{L_1,...,L_n\}\cup$\\ $\{comp(L_1),...,comp(L_n)\})\cap \{M_{i,j}~|~i\leq n\} = \emptyset$ follows directly from the definition of a trail induced ordering.
\ref{appr_4:trail_ind_order}.8 and \ref{appr_4:trail_ind_order}.9 are disjunct since $\{L~|~L~is~of~level~0\}$ and $\{L~|~L~is~of~level~greater~0\}$ are disjunct by lemma~\ref{gammastar:properties:proofs:litlvlunique}.
It follows that \ref{appr_4:trail_ind_order}.1-\ref{appr_4:trail_ind_order}.9 are pairwise disjunct, since all relations in \ref{appr_4:trail_ind_order}.1-\ref{appr_4:trail_ind_order}.7 contain only $\beta\mhyphen \mathit{defined}$ literals of level 1 or higher
and all relations in \ref{appr_4:trail_ind_order}.8, \ref{appr_4:trail_ind_order}.9 contain at least one $\beta\mhyphen \mathit{defined}$ literal of level 0.
\ref{appr_4:trail_ind_order}.10 and \ref{appr_4:trail_ind_order}.11 are disjunct since a literal cannot be both $\beta\mhyphen \mathit{defined}$ and $\beta\mhyphen \mathit{undefined}$. It follows that \ref{appr_4:trail_ind_order}.1-\ref{appr_4:trail_ind_order}.11 are pairwise disjunct, since all relations in \ref{appr_4:trail_ind_order}.1-\ref{appr_4:trail_ind_order}.9 contain only $\beta\mhyphen \mathit{defined}$ literals
and all relations in \ref{appr_4:trail_ind_order}.10, \ref{appr_4:trail_ind_order}.11 contain at least one $\beta\mhyphen \mathit{undefined}$ literal.
\end{proof}

\subsubsection{Proof of Lemma~\ref{gammastar:properties}-\ref{gammastar:properties:strictorder}}
$\prec_{\Gamma^*}$ is a total strict order, i.e. $\prec_{\Gamma^*}$ is irreflexive, transitive and total.
\begin{proof}
  Suppose a trail $\Gamma :=[L_1^{i_1:C_1\cdotp \sigma_1},...,L_n^{i_n:C_n\cdotp \sigma_n}]$ and a term $\beta$ such that $\{L_1,...,L_n\}\prec_T\beta$.\\
  \emph{Irreflexivity}.   We have to show that there is no ground literal $L$ such that $L\prec_{\Gamma^*} L$. Suppose two literals $L$ and $K$ such that $L\prec_{\Gamma^*} K$ and $L=K$. Now we have several cases:
  \begin{enumerate}
    \item Suppose that $L,K$ are $\beta\mhyphen\mathit{defined}$ and of level $1$ or higher. Then we have several cases:
    \begin{enumerate}
      \item $L=M_{i,j}$ and $K=M_{k,l}$. Then by \ref{appr_4:trail_ind_order}.1 $M_{i,j}\prec_{\Gamma^*} M_{k,l}$ if $i<k$ or $(i=k$ and $j<l)$. Thus $i\not=k$ or $j\not=l$.
      We show that for $M_{i,j}$, $M_{k,l}$ with $i\not=k$ or $j\not=l$ it holds $M_{i,j}\not=M_{k,l}$. Assume  that $M_{i,j} = M_{k,l}$ and $k\not=i$ or $j\not= l$. Assume that $k=i$.
      Then, by definition \ref{appr_4:trail_ind_order} $M_{i,j} \prec_T M_{k,l}$ or $M_{k,l} \prec_T M_{i,j}$. Thus $M_{i,j} \not= M_{k,l}$ since $\prec_T$ is a rewrite ordering.
      Now assume that $k\not=i$. Since $M_{i,j} = M_{k,l}$ it holds $max_\Gamma(M_{i,j}) = max_\Gamma(M_{k,l})$, since both have the same level by lemma~\ref{gammastar:properties:proofs:litlvlunique}.
      But then $k=i$. Thus $M_{i,j} \not= M_{k,l}$ for $k\not=i$ or $j\not= l$. Thus if by \ref{appr_4:trail_ind_order}.1 $M_{i,j}\prec_{\Gamma^*} M_{k,l}$ if $i<k$ or ($i=k$ and $j<l$), then $M_{i,j}\not= M_{k,l}$.
      \item $L=L_i$ and $K=L_j$. Then by \ref{appr_4:trail_ind_order}.2 $L_i\prec_{\Gamma^*} L_j$ if $L_i\prec_\Gamma L_j$.
      Then by lemma~\ref{gammastar:properties:proofs:uniqueness} $L_i\not=L_j$.
      \item $L=comp(L_i)$ and $K=L_j$. Then by \ref{appr_4:trail_ind_order}.3 $comp(L_i)\prec_{\Gamma^*} L_j$ if $L_i\prec_\Gamma L_j$.
      $L_i\not=L_j$ by lemma~\ref{gammastar:properties:proofs:uniqueness}. $L\not= K$ has to hold since $\Gamma$ is consistent.
      \item $L=L_i$ and $K=comp(L_j)$. Then by \ref{appr_4:trail_ind_order}.4 $L_i\prec_{\Gamma^*} comp(L_j)$ if $L_i\prec_\Gamma L_j$ or $i=j$. If $i\not=j$ then we can proceed analogous to the previous step. If $i=j$ then obviously $L_i\not=comp(L_i)$.
      \item $L=comp(L_i)$ and $K=comp(L_j)$. Then by \ref{appr_4:trail_ind_order}.5 $comp(L_i)\prec_{\Gamma^*} comp(L_j)$ if $L_i\prec_\Gamma L_j$.
      By lemma~\ref{gammastar:properties:proofs:uniqueness} $L_i\not= L_j$. Thus $comp(L_i)\not= comp(L_j)$.
      \item $L=L_i$ and $K=M_{k,l}$. Then by \ref{appr_4:trail_ind_order}.6 $L_i\prec_{\Gamma^*} M_{k,l}$,  $\comp(L_i)\prec_{\Gamma^*} M_{k,l}$ if $i\leq k$.
      $M_{k,l}\not=L_i$ and $M_{k,l}\not=comp(L_i)$ follows directly from the definition \ref{appr_4:trail_ind_order}.
      Thus if $L_i\prec_{\Gamma^*} M_{k,l}$ or  $\comp(L_i)\prec_{\Gamma^*} M_{k,l}$ if $i\leq k$ by \ref{appr_4:trail_ind_order}.6, then $L_i\not= M_{k,l}$ and $\comp(L_i)\not= M_{k,l}$.
      \item $L=M_{k,l}$ and $K=L_i$. Then we can proceed analogous to the previous step for \ref{appr_4:trail_ind_order}.7.
    \end{enumerate}
      \item Suppose that $L$ and $K$ are $\beta\mhyphen\mathit{defined}$ and of level zero. Since $\prec_T$ is irreflexive, $L\not\prec_T K$ has to hold. Since $\prec_{\Gamma^*} = \prec_T$ for literals of level zero $L\not\prec_{\Gamma^*} K$ has to hold too.
      \item Suppose that $L,K$ are $\beta\mhyphen\mathit{defined}$ and $L$ is of level zero and $K$ is of level greater than zero. But then $L\not=K$ has to hold by lemma~\ref{gammastar:properties:proofs:litlvlunique}. Thus $L\not\prec_{\Gamma^*} K$ for \ref{appr_4:trail_ind_order}.9.
      \item Suppose that $L$ and $K$ are $\beta\mhyphen\mathit{undefined}$. Then by \ref{appr_4:trail_ind_order}.10 $K\prec_{\Gamma^*} H$ if $K\prec_T H$. Since $\prec_T$ is a rewrite ordering $K\prec_T H$ iff $K\not= H$.
      \item Suppose that $L$ is $\beta\mhyphen\mathit{defined}$ and $K$ is $\beta\mhyphen\mathit{undefined}$. Then by \ref{appr_4:trail_ind_order}.11 $L\prec_{\Gamma^*} K$. Then $L\not=K$ has to hold since otherwise $L,K$ would be both $\beta\mhyphen\mathit{defined}$ and $\beta\mhyphen\mathit{undefined}$, contradicting consistency of $\Gamma$.
    \end{enumerate}
  \emph{Transitivity}. Suppose there exist literals $L, K, H$ such that $H\prec_{\Gamma^*} K$ and $K \prec_{\Gamma^*} L$ but not $H\prec_{\Gamma^*} L$. We have several cases:
  \begin{enumerate}
    \item Suppose all literals are $\beta\mhyphen\mathit{undefined}$. Then $K \prec_T L$ and $H\prec_T K$. Otherwise $K \prec_{\Gamma^*} L$ and $H\prec_{\Gamma^*} K$ would not hold. Thus also $H\prec_T L$ by transitivity of $\prec_T$. Thus $H\prec_{\Gamma^*} L$ by \ref{appr_4:trail_ind_order}.10.
    \item Suppose two literals are $\beta\mhyphen\mathit{undefined}$. If $K$ would be $\beta\mhyphen\mathit{defined}$, then $K\prec_{\Gamma^*} H$ by \ref{appr_4:trail_ind_order}.11 contradicting assumption. If $L$ would be $\beta\mhyphen\mathit{defined}$, then $L \prec_{\Gamma^*} K$ by \ref{appr_4:trail_ind_order}.11 again contradicting assumption. Thus $H$ has to be $\beta\mhyphen\mathit{defined}$.
    Then $H \prec_{\Gamma^*} L$ by definition \ref{appr_4:trail_ind_order}.11.
    \item Suppose one literal is $\beta\mhyphen\mathit{undefined}$. If $K$ would be $\beta\mhyphen\mathit{undefined}$, then $L \prec_{\Gamma^*} K$ by definition \ref{appr_4:trail_ind_order}.11 contradicting assumption. If $H$ would be $\beta\mhyphen\mathit{undefined}$, then $K \prec_{\Gamma^*} H$ by definition \ref{appr_4:trail_ind_order}.11 again contradicting assumption. Thus $L$ has to be $\beta\mhyphen\mathit{undefined}$.
    Then $H \prec_{\Gamma^*} L$ by definition \ref{appr_4:trail_ind_order}.11.
    \item Suppose all literals are $\beta\mhyphen\mathit{defined}$. Then we have multiple subcases:
    \begin{enumerate}
      \item Suppose all literals have the same defining literal $L_i$ and $L_i$ is of level 1 or higher. By \ref{appr_4:trail_ind_order}.6 $L_i\prec_{\Gamma^*} M_{i,j}$ and $comp(L_i)\prec_{\Gamma^*} M_{i,j}$ for all $j$.  By \ref{appr_4:trail_ind_order}.4 $L_i\prec_{\Gamma^*} comp(L_i)$.
      Thus $L_i\prec_{\Gamma^*} comp(L_i)\prec_{\Gamma^*} M_{i,j}$ for all $j$.
      Since $K \prec_{\Gamma^*} L$ either $K=L_i$ and $L\not=L_i$ or $L=M_{i,j}$ and $K=comp(L_i)$ or $L=M_{i,j}$ and $K=M_{i,k}$ with $k<j$.
      \begin{enumerate}
      \item Assume $K=L_i$ and $L\not=L_i$.  Since $K$ is the smallest literal with defining literal $L_i$, $K = H$ has to hold. But then $K \prec_{\Gamma^*} K$ contradicting irreflexivity.
      \item Assume $L=M_{i,j}$ and $K=comp(L_i)$. Since $H \prec_{\Gamma^*} K$ and all literals have the same defining literal, $H = L_i$ has to hold by \ref{appr_4:trail_ind_order}.6 and \ref{appr_4:trail_ind_order}.4.
      Then, again by \ref{appr_4:trail_ind_order}.6, $H \prec_{\Gamma^*} L$.
      \item $L=M_{i,j}$ and $K=M_{i,k}$ with $k<j$.
      Since $H \prec_{\Gamma^*} K$ and all have the same defining literal either $H = M_{i,l}$ with $l<k$ by \ref{appr_4:trail_ind_order}.1, or $H=L_i$ or $H=comp(L_i)$ by \ref{appr_4:trail_ind_order}.6. In both cases $H \prec_{\Gamma^*} L$ holds by \ref{appr_4:trail_ind_order}.1 and \ref{appr_4:trail_ind_order}.6.
      \end{enumerate}
      \item Suppose $H,K,L$ have at least one different defining literal and $max_\Gamma(L) = L_i$ with $L_i$ of level 1 or higher. First, we have to show that if  $L_j = max_\Gamma(K') \prec_\Gamma max_\Gamma(L') = L_i$ and $L_i$ is of level 1 or higher, then $K' \prec_{\Gamma^*} L'$.
      Suppose that $L_j$ is of level 0. Then $K' \prec_{\Gamma^*} L'$ by \ref{appr_4:trail_ind_order}.9.
      Suppose that $L'= M_{i,k}$ and $K'= M_{j,l}$. Then $K' \prec_{\Gamma^*} L'$ by \ref{appr_4:trail_ind_order}.1.
      Suppose that $L'= M_{i,k}$ and $K'=L_j$ or $K'=comp(L_j)$. Then $K' \prec_{\Gamma^*} L'$ by \ref{appr_4:trail_ind_order}.6.
      Suppose that $L'= L_i$ or $L'=comp(L_i)$ and $K'=L_j$ or $K'=comp(L_j)$. Then $K' \prec_{\Gamma^*} L'$ by \ref{appr_4:trail_ind_order}.2-\ref{appr_4:trail_ind_order}.5.
      Suppose that $L'= L_i$ or $L'=comp(L_i)$ and $K'= M_{j,l}$. Then $K' \prec_{\Gamma^*} L'$ by \ref{appr_4:trail_ind_order}.7.\\
      Now by assumption $H \prec_{\Gamma^*} K$ and $K \prec_{\Gamma^*} L$. If  $max_\Gamma(K) \prec_\Gamma max_\Gamma(H)$ then $K \prec_{\Gamma^*} H$ contradicting assumption. The same holds for $L$ and $K$.
      Thus either $max_\Gamma(H) \prec_\Gamma max_\Gamma(L)$ or $max_\Gamma(K) \prec_\Gamma max_\Gamma(L)$. In the first case $H \prec_{\Gamma^*} L$ follows from above. In the second case $max_\Gamma(H) \preceq_\Gamma max_\Gamma(K)$ has to hold. Thus $H \prec_{\Gamma^*} L$ follows again.
      \item Suppose that $max_\Gamma(L) = L_i$ where $L_i$ is of level 0. Since $K \prec_{\Gamma^*} L$, $max_\Gamma(K) = L_j$ with $L_j$ of level 0 has to hold by \ref{appr_4:trail_ind_order}.9 and \ref{appr_4:trail_ind_order}.11.
      Now assume that $L\prec_T K$. Then $L\prec_{\Gamma^*} K$ by \ref{appr_4:trail_ind_order}.8 contradicting assumption. Thus $K\prec_T L$ has to hold since $K\not= L$.
      Since $H \prec_{\Gamma^*} K$, $max_\Gamma(H) = L_k$ with $L_k$ of level 0 has to hold by \ref{appr_4:trail_ind_order}.9 and \ref{appr_4:trail_ind_order}.11.
      Now assume that $K\prec_T H$. Then $K\prec_{\Gamma^*} H$ by \ref{appr_4:trail_ind_order}.8 contradicting assumption. Thus $H\prec_T K$ has to hold since $H\not= K$.
      By transitivity of $\prec_T$, $H\prec_T L$ and thus $H \prec_{\Gamma^*} L$ has to hold.
    \end{enumerate}
  \end{enumerate}
  \emph{Totality}.
  First we show that any ground literal is either $\beta\mhyphen\mathit{defined}$ and has a level or $\beta\mhyphen\mathit{undefined}$. Since $\Gamma$ is consistent, a literal is either $\beta\mhyphen\mathit{defined}$ or $\beta\mhyphen\mathit{undefined}$. We just need to show that if a literal is $\beta\mhyphen\mathit{defined}$, it has a level.
  By lemma \ref{gammastar:properties:proofs:alldefhavelevel} all defined literals have a level. $\beta\mhyphen\mathit{definedness}$ implies definedness. Thus all $\beta\mhyphen\mathit{defined}$ literals have a level. Now assume some arbitrary ground literals $L\neq K$. We have several cases:
  \begin{enumerate}
    \item $L,K$ are $\beta\mhyphen\mathit{undefined}$. Since $L\not=K$ we have $L\prec_T K$ or $K\prec_T L$ by totality of $\prec_T$ on ground literals. Thus by \ref{appr_4:trail_ind_order}.10 $L\prec_{\Gamma^*} K$ or $K\prec_{\Gamma^*} L$.
    \item One is $\beta\mhyphen\mathit{defined}$. Then either $L\prec_{\Gamma^*} K$ or $K\prec_{\Gamma^*} L$ by \ref{appr_4:trail_ind_order}.11.
    \item Both are $\beta\mhyphen\mathit{defined}$. Then we have several subcases:
    \begin{enumerate}
      \item $L$ is of level zero and $K$ is of level greater than zero or vice versa. Then by \ref{appr_4:trail_ind_order}.9 $L\prec_{\Gamma^*} K$ or $K\prec_{\Gamma^*} L$ has to hold.
      \item $max_\Gamma(L)= L_i$ and $max_\Gamma(K)= L_j$ and both are of level 1 or higher.
      \begin{enumerate}
        \item $L=M_{i,k}$ and $K=M_{j,l}$. Then either $M_{i,k}\prec_{\Gamma^*} M_{j,l}$ or $M_{j,l}\prec_{\Gamma^*} M_{i,k}$ by \ref{appr_4:trail_ind_order}.1.
        \item $L=M_{i,k}$ and $K=L_j$ or $K=comp(L_j)$. If $i\geq j$ then by \ref{appr_4:trail_ind_order}.6 $L_j\prec_{\Gamma^*} M_{i,k}$ or $ comp(L_j)\prec_{\Gamma^*} M_{i,k}$.
        If $i<j$ then by \ref{appr_4:trail_ind_order}.7 $M_{i,k}\prec_{\Gamma^*} L_j$ or $M_{i,k}\prec_{\Gamma^*} comp(L_j)$.
        \item $K=M_{i,k}$ and $L=L_j$ or $L=comp(L_j)$. Analogous to previous step.
        \item $L=L_i$ and $K=L_j$. Then if $i<j$ by \ref{appr_4:trail_ind_order}.2 $L_i\prec_{\Gamma^*} L_j$ and $L_j\prec_{\Gamma^*} L_i$ otherwise.
        \item $L=comp(L_i)$ and $K=comp(L_j)$ analogous to previous step for \ref{appr_4:trail_ind_order}.5.
        \item $L=L_i$ and $K=comp(L_j)$. Then if $i\leq j$ by \ref{appr_4:trail_ind_order}.4 $L_i\prec_{\Gamma^*} comp(L_j)$. If $j<i$ by \ref{appr_4:trail_ind_order}.3 $comp(L_j)\prec_{\Gamma^*} L_i$.
        \item $L=comp(L_i)$ and $K=L_j$. Then if $i< j$ by \ref{appr_4:trail_ind_order}.3 $comp(L_i)\prec_{\Gamma^*} L_j$. If $j\leq i$ by \ref{appr_4:trail_ind_order}.4 $L_j\prec_{\Gamma^*} comp(L_i)$.
      \end{enumerate}
      \item $max_\Gamma(L)= L_i$ and $max_\Gamma(K)= L_j$ and both are of level 0. Now either $L\prec_T K$ or $K\prec_T L$. Thus by \ref{appr_4:trail_ind_order}.8 $L\prec_{\Gamma^*} K$ or $K\prec_{\Gamma^*} L$.
    \end{enumerate}
  \end{enumerate}
\end{proof}

\subsubsection{Proof of Lemma~\ref{gammastar:properties}-\ref{gammastar:properties:wellfounded}}
$\prec_{\Gamma^*}$ is a well-founded ordering.
\begin{proof}
 Suppose some arbitrary subset $M$ of all ground literals, a trail $\Gamma:=[L_1^{i_1:C_1\cdotp \sigma_1},...,L_n^{i_n:C_n\cdotp \sigma_n}]$ and a term $\beta$ such that $\{L_1,...,L_n\}\prec_T \beta$. We have to show that $M$ has a minimal element. We have several cases:
 \begin{enumerate}
   \item $L$ is $\beta\mhyphen\mathit{undefined}$ in $\Gamma$ for all literals $L\in M$. Then $\prec_{\Gamma^*} = \prec_T$. Since $\prec_T$ is well-founded there exists a minimal element in $M$. Thus there exists a minimal element in $M$ with regard to $\prec_{\Gamma^*}$.
   \item there exists at least one literal in $M$ that is $\beta\mhyphen\mathit{defined}$. Then we have two cases:
   \begin{enumerate}
     \item there exists a literal in $M$ that is of level zero. Then let $L\in M$ be the literal of level zero, where $L\prec_T K$ for all $K\in M$ with $K$ of level zero. We show that $L$ is the minimal element. Suppose there exists a literal $L'\in M$ that is smaller than $L$.
     Since $L$ is of level zero, $L\prec_{\Gamma^*} K$ for all literals $K$ of level greater than zero by \ref{appr_4:trail_ind_order}.9 and $L\prec_{\Gamma^*} H$ for all $\beta\mhyphen\mathit{undefined}$ literals  $H$ by \ref{appr_4:trail_ind_order}.11. Thus $L'$ must be of level zero. But then $L'\prec_T L$ has to hold, contradicting assumption.
     \item There exists no literal in $M$ that is of level zero.
     Let $L\in M$ be the literal where $max_\Gamma(L)\preceq_\Gamma max_\Gamma(K)$ for all $K\in M$ and
    \begin{enumerate}
      \item $L= max_\Gamma(L)$ or
      \item $L= comp(max_\Gamma(L))$ and $max_\Gamma(L)\not\in M$ or
      \item $L\prec_T H$ for all $H\in M$ such that $max_\Gamma(L)= max_\Gamma(H)$  and $max_\Gamma(L)\not\in M$ and $comp(max_\Gamma(L))\not\in M$.
    \end{enumerate}
     We show that $L$ is the minimal element. Suppose there exists a literal $L'\in M$ that is smaller than $L$. We have three cases:
     \begin{enumerate}
       \item $max_\Gamma(L) = L = L_i$. Since $L_i\prec_T \beta$ we have either $L' = L_j$ with $j < i$ by \ref{appr_4:trail_ind_order}.2 or $L'= comp(L_j)$ with $j< i$  by \ref{appr_4:trail_ind_order}.3 or $L' = M_{k,l}$ with $k < i$  by \ref{appr_4:trail_ind_order}.7.
       In all three cases we have $max_\Gamma(L')\prec_\Gamma max_\Gamma(L)$ contradicting assumption that the defining literal of $L$ is minimal in $M$.
       \item $L=comp(L_i) = comp(max_\Gamma(L))$ and $max_\Gamma(L)\not\in M$.
       Since $L_i\prec_T\beta$ either $L' = L_j$ with $j < i$ by \ref{appr_4:trail_ind_order}.4 or $L'= comp(L_j)$ with $j< i$  by \ref{appr_4:trail_ind_order}.5 or $L' = M_{k,l}$ with $k < i$  by \ref{appr_4:trail_ind_order}.7.
       In all three cases we have $max_\Gamma(L')\prec_\Gamma max_\Gamma(L)$ contradicting assumption that the defining literal of $L$ is minimal in $M$.
       \item $L = M_{k,l}$ and $max_\Gamma(L)\not\in M$ and $comp(max_\Gamma(L))\not\in M$. Then either $L'=M_{i,j}$ with $i<k$ or $(i=k$ and $j<l)$ by \ref{appr_4:trail_ind_order}.1 or $L'=L_i$ or $L'=comp(L_i)$ with $i\leq k$.
       Suppose that  $L'=M_{i,j}$ and $i<k$. Then $max_\Gamma(L')\prec_\Gamma max_\Gamma(L)$ contradicting assumption.
       Suppose that  $L'=M_{i,j}$ and $i=k$ and $j<l$. Then $L'\prec_T L$ and $max_\Gamma(L)= max_\Gamma(L')$. For $L$ it holds $L\prec_T H$ for all $H\in M$ such that $max_\Gamma(L)= max_\Gamma(H)$. Contradiction.
       Suppose that $L'=L_i$ or $L'=comp(L_i)$ with $i = k$. Then  $max_\Gamma(L) = L_i$.
       By assumption $max_\Gamma(L)\not\in M$ and $comp(max_\Gamma(L))\not\in M$. Contradiction.
       Suppose that $L'=L_i$ or $L'=comp(L_i)$ with $i < k$. Then we have $max_\Gamma(L')\prec_\Gamma max_\Gamma(L)$ contradicting assumption that the defining literal of $L$ is minimal in $M$.
     \end{enumerate}

   \end{enumerate}
 \end{enumerate}
\end{proof}

\subsubsection{Proof of Lemma~\ref{lem:np-complete}}
  Assume that all ground terms $t$ with $t\prec_T \beta$ for any $\beta$ are polynomial in the size of $\beta$.
  Then testing Propagate (Conflict) is NP-Complete, i.e., the problem of checking for a given clause $C$ whether there exists a grounding
  substitution $\sigma$ such that $C\sigma$ propagates (is false) is NP-Complete.
\begin{proof}
  Let $C\sigma$ be propagable (false). The problem is in NP because $\beta$ is constant and for all $t\in codom(\sigma)$ it holds that $t$ is polynomial in the size of $\beta$.
  Checking if $C\sigma$ is propagable (false) can be done in polynomial time with Congruence Closure~\cite{NelsonOppen80} since $\sigma$ has polynomial size.

We reduce 3-SAT to testing rule Conflict. Consider a 3-place predicate $R$, a unary function $g$, and a mapping
from propositional variables $P$ to first-order variables $x_P$. Assume a 3-SAT clause set $N =\{ \{L_0,L_1,L_2\},...,\{L_{n-2},L_{n-1},L_{n}\}\}$, where $L_i$ may denote both $P_i$ and $\neg P_i$.
Now we create the clause $$\{R(t_{0},t_{1},t_{2})\not\approx true,...,R(t_{n-2},t_{n-1},t_{n}\not\approx true)\}$$ where $t_{i} := x_{P_{i}}$ if $L_{i} = P_{i}$ and $t_{i} := g(x_{P_{i}})$ otherwise.
Now let $\Gamma:=\{R(x_0,x_1,x_2)\mid x_i\in\{0,1,g(0),g(1)\}$ such that $(x_0\lor x_1 \lor x_2)\downarrow_{\{g(x)\mapsto (\neg x)\}}$ is true $\}$ be the set of all $R$-atoms that evaluate to true
if considered as a three literal propositional clause. Now $N$ is satisfiable
if and only if Conflict is applicable to the new clause. The reduction is analogous for Propagate.
\end{proof}

\subsubsection{Proof of Theorem~\ref{appr_4:scleq_sound}}
Assume a state $(\Gamma;N;U;\beta;k;D)$ resulting from a run. Then $(\Gamma;N;U;\beta;k;D)$ is sound.
\begin{proof}
 Proof by structural induction on $(\Gamma;N;U;\beta;k;D)$. Let $(\Gamma;N;U;\beta;k; D) = (\epsilon, N, \emptyset, \beta, 0, \top)$, the initial state. Then it is sound according to lemma \ref{appr_4:init_state_sound}. Now assume that $(\Gamma;N;U;\beta;k;D)$ is sound. We need to show that any application of a rule results in a sound state.
\paragraph{Propagate:} Assume $\mathit{Propagate}$ is applicable. Then there exists $C\in N\cup U$ such that $C=C_0\lor C_1\lor L$, $L\sigma$ is $\beta\mhyphen\mathit{undefined}$ in $\Gamma$, $C_1\sigma = L\sigma\lor ...\lor L\sigma$, $C_1=L_1\lor ...\lor L_n$,$\mu=mgu(L_1,...,L_n,L)$ and $C_0\sigma$ is $\beta\mhyphen\mathit{false}$ in $\Gamma$.
       Then a reduction chain application $[I_1,...,I_m]$ from $\Gamma$ to $L\sigma^{k:(C_0\lor L)\mu\cdotp \sigma}$ is created with $I_m := (s_m\# t_m\cdotp\sigma_m, s_m\# t_m\lor C_m\cdotp\sigma_m, I_j, I_k, p_m)$.
       Finally $s_m\# t_m\sigma_m^{k:(s_m\# t_m\lor C_m)\cdotp\sigma_m}$ is added to $\Gamma$.\\
       By definition of a reduction chain application $(s_m\# t_m)\sigma_m = L\sigma{\downarrow_{conv(\Gamma)}}$. Thus, $(s_m\# t_m)\sigma_m$ must be $\beta\mhyphen\mathit{undefined}$ in $\Gamma$ and irreducible by $conv(\Gamma)$, since $(C_0\lor L)\mu\sigma\prec_T \beta$ by definition of Propagate.
    \begin{itemize}
    \item \ref{appr_4:soundness}.\ref{appr_4:sound_1}:
    Since $(s_m\# t_m)\sigma_m$ is $\beta\mhyphen\mathit{undefined}$ in $\Gamma$, adding $(s_m\# t_m)\sigma_m$ does not make $\Gamma$ inconsistent. Thus $\Gamma,(s_m\# t_m)\sigma_m$ remains consistent.
    \item \ref{appr_4:soundness}.\ref{appr_4:sound_2}: $(s_m\# t_m)\sigma_m$ is $\beta\mhyphen\mathit{undefined}$ in $\Gamma$ and irreducible by $conv(\Gamma)$.
    It remains to show that $C_m\sigma_m$ is $\beta\mhyphen\mathit{false}$ in $\Gamma$, $N \cup U \models s_m\# t_m\lor C_m$ and $(s_m\# t_m\lor C_m)\sigma_m\prec_T \beta$.
    By i.h. for all $L'\sigma'^{l:(L'\lor C')\cdotp \sigma'}\in\Gamma$ it holds that $C'\sigma'$ is $\beta\mhyphen\mathit{false}$ in $\Gamma$, $(L'\lor C')\sigma'\prec_T \beta$ and $N\cup U \models (L'\lor C')$. By definition of Propagate $C_0\sigma$ is $\beta\mhyphen\mathit{false}$ in $\Gamma$ and $C\sigma\prec_T \beta$ and $N\cup U \models C$.
    $(C_0\lor C_1\lor L)\mu$ is an instance of $C$. Thus $C \models (C_0\lor C_1\lor L)\mu$. $C_0\mu = L\mu \lor...\lor L\mu$ by definition of Propagate. Thus $C \models (C_1\lor L)\mu$ and by this $N\cup U \models (C_1\lor L)\mu$.
    By definition of a reduction chain application $I_j$ either contains a clause annotation from $\Gamma,L\sigma^{k:(C_0\lor L)\cdotp \sigma}$ or it is a rewriting inference from smaller rewrite steps for all $1\leq j\leq m$.
    Thus, by lemma \ref{rewrite_inference_properties} it follows by induction that for any rewriting inference $I_j := (s_j\# t_j\cdotp\sigma_j, s_j\# t_j\lor C_j\cdotp\sigma_j, I_i, I_k, p_j)$ it holds $C_j\sigma_j$ is $\beta\mhyphen\mathit{false}$ in $\Gamma$, $N \cup U \models s_j\# t_j\lor C_j$ and $(s_j\# t_j\lor C_j)\sigma_j\prec_T \beta$.
    \item \ref{appr_4:soundness}.\ref{appr_4:sound_3} and \ref{appr_4:soundness}.\ref{appr_4:sound_4} trivially hold by induction hypothesis.
    \item \ref{appr_4:soundness}.\ref{appr_4:sound_5}: trivially holds since $D = \top$.
   \end{itemize}
\paragraph{Decide:} Assume $Decide$ is applicable. Then there exists $C\in N\cup U$ such that $C=C_0\lor L$, $L\sigma$ is ground and $\beta\mhyphen\mathit{undefined}$ in $\Gamma$ and $C_0\sigma$ is ground and $\beta\mhyphen\mathit{undefined}$ or $\beta\mhyphen\mathit{true}$ in $\Gamma$.
   Then a reduction chain application $[I_1,...,I_m]$ from $\Gamma$ to $L\sigma^{k+1:(C_0\lor L)\cdotp \sigma}$ is created with $I_m := (s_m\# t_m\cdotp\sigma_m, s_m\# t_m\lor C_m\cdotp\sigma_m, I_j, I_k, p_m)$.
   Finally $s_m\# t_m\sigma_m^{k+1:(s_m\# t_m\lor comp(s_m\# t_m))\cdotp\sigma_m}$ is added to $\Gamma$.\\
   By definition of a reduction chain application $(s_m\# t_m)\sigma_m = L\sigma{\downarrow_{conv(\Gamma)}}$. Thus, $(s_m\# t_m)\sigma_m$ must be $\beta\mhyphen\mathit{undefined}$ in $\Gamma$ and irreducible by $conv(\Gamma)$, since $(C_0\lor L)\sigma\prec_T \beta$ by definition of Decide.
   \begin{itemize}
     \item \ref{appr_4:soundness}.\ref{appr_4:sound_1}:
     Since $(s_m\,\#\, t_m)\sigma_m$ is $\beta\mhyphen\mathit{undefined}$ in $\Gamma$ adding $(s_m\# t_m)\sigma_m$ does not make $\Gamma$ inconsistent. Thus $\Gamma,(s_m\# t_m)\sigma_m$ remains consistent.
     \item \ref{appr_4:soundness}.\ref{appr_4:sound_3}: $(s_m\# t_m)\sigma_m$ is $\beta\mhyphen\mathit{undefined}$ in $\Gamma$ and irreducible by $conv(\Gamma)$. $N\cup U\models (s_m\# t_m)\lor comp(s_m\# t_m)$ obviously holds.
     $(s_m\# t_m)\sigma_m\prec_T \beta$ holds inductively by lemma \ref{rewrite_inference_properties} and since $L\sigma\prec_T \beta$.
     \item \ref{appr_4:soundness}.\ref{appr_4:sound_2} and \ref{appr_4:soundness}.\ref{appr_4:sound_4} trivially hold by induction hypothesis.
     \item \ref{appr_4:soundness}.\ref{appr_4:sound_5}: trivially holds since $D = \top$.
   \end{itemize}

\paragraph{Conflict:}  Assume $\mathit{Conflict}$ is applicable. Then there exists a $D'\sigma$ such that $D'\sigma$ is $\beta\mhyphen\mathit{false}$ in $\Gamma$. Then:
    \begin{itemize}
     \item \ref{appr_4:soundness}.\ref{appr_4:sound_1} - \ref{appr_4:soundness}.\ref{appr_4:sound_4} trivially hold by induction hypothesis
     \item \ref{appr_4:soundness}.\ref{appr_4:sound_5}: $D'\sigma$ is $\beta\mhyphen\mathit{false}$ in $\Gamma$ by definition of $\mathit{Conflict}$. Now we have two cases:
     \begin{enumerate}
     \item $D'\sigma$ is of level greater than zero. Then $N \cup U \models D'$ since $D'\in N\cup U$ by definition of $\mathit{Conflict}$.
     \item $D'\sigma$ is of level zero. Then we have to show that $N \cup U \models \bot$. For any literal $L_0^{0:(L_0\lor D_0)\cdotp \sigma}\in \Gamma$ it holds $N\models L_0$, since any literal of level $0$ is a propagated literal. By definition of a level, for any $K\in D'\sigma$ there exists a core $core(\Gamma;K)$ that contains only literals of level $0$.
     Thus $N\cup U\models core(\Gamma;K)$ and $core(\Gamma;K)\models \neg K$ for any such $K$. Then $N\cup U\models \neg D'\sigma$ and $N\cup U\models D'\sigma$ and therefore $N\cup U\models \bot$.
     \end{enumerate}
    \end{itemize}

\paragraph{Skip:} Assume $\mathit{Skip}$ is applicable. Then $\Gamma = \Gamma', L$ and $D = D'\clsr\sigma$ and $D'\sigma$ is $\beta\mhyphen\mathit{false}$ in $\Gamma'$.
    \begin{itemize}
     \item \ref{appr_4:soundness}.\ref{appr_4:sound_1}: By i.h. $\Gamma$ is consistent. Thus $\Gamma'$ is consistent as well.
     \item \ref{appr_4:soundness}.\ref{appr_4:sound_2}- \ref{appr_4:soundness}.\ref{appr_4:sound_4}: trivially hold by induction hypothesis and since $\Gamma'$ is a prefix of $\Gamma$.
     \item \ref{appr_4:soundness}.\ref{appr_4:sound_5}: By i.h. $D'\sigma$ is $\beta\mhyphen\mathit{false}$ in $\Gamma$ and $N\cup U\models D'$. By definition of $\mathit{Skip}$ $D'\sigma$ is $\beta\mhyphen\mathit{false}$ in $\Gamma'$.
    \end{itemize}

\paragraph{Explore-Refutation:} Assume $\mathit{Explore\mhyphen Refutation}$ is applicable. Then $D=(D'\lor s\,\#\,t)\cdotp\sigma$, $(s\,\#\,t)\sigma$ is strictly $\prec_{\Gamma^*}$ maximal in $(D'\lor s\,\#\,t)\sigma$,
      $[I_1,...,I_m]$ is a refutation from $\Gamma$ and $(s\,\#\,t)\sigma$,
      $I_j = (s_j\# t_j\cdotp\sigma_j, (s_j\# t_j\lor C_j)\cdotp\sigma_j, I_l, I_k, p_j)$, $1\leq j\leq m$,
      $ (s_j\,\#\,t_j\lor C_j)\sigma_j \prec_{\Gamma^*} (D'\lor s\,\#\,t)\sigma$, $(s_j\# t_j\lor C_j)\sigma_j$ is $\beta\mhyphen\mathit{false}$ in $\Gamma$.
      \begin{itemize}
       \item \ref{appr_4:soundness}.\ref{appr_4:sound_1}-\ref{appr_4:soundness}.\ref{appr_4:sound_4} trivially hold by i.h.
       \item \ref{appr_4:soundness}.\ref{appr_4:sound_5}. By definition $(C_j\lor s_j\,\#\,t_j)\sigma_j$ is $\beta\mhyphen\mathit{false}$ in $\Gamma$.
       By i.h. for all $L'\sigma'^{l:(L'\lor C')\cdotp \sigma'}$ $\in \Gamma$ it holds that $N\cup U \models (L'\lor C')$. By i.h. $N\cup U \models D'\lor s\,\#\,t$.
       By definition of a refutation $I_j := (s_j\# t_j\cdotp\sigma_j, s_j\# t_j\lor C_j\cdotp\sigma_j, I_i, I_k, p_j)$ either contains a clause annotation from $\Gamma,(s\,\#\,t)\sigma^{k:(D'\lor s\,\#\,t)\cdotp \sigma}$ or it is a rewriting inference from smaller rewrite steps for all $1\leq j\leq m$.
       Thus it follows inductively by lemma \ref{rewrite_inference_properties} that $N\cup U \models (s_j\# t_j\lor C_j)$.
     \end{itemize}

\paragraph{Factorize:}  Assume $\mathit{Factorize}$ is applicable. Then $D=D'\cdotp \sigma$.
    \begin{itemize}
     \item \ref{appr_4:soundness}.\ref{appr_4:sound_1} - \ref{appr_4:soundness}.\ref{appr_4:sound_4} trivially hold by induction hypothesis.
     \item \ref{appr_4:soundness}.\ref{appr_4:sound_5}: By i.h. $D'\sigma$ is $\beta\mhyphen\mathit{false}$ in $\Gamma$ and $N\cup U\models D'$.
     By the definition of $\mathit{Factorize}$ $D' = D_0 \lor L \lor L'$ such that $L\sigma=L'\sigma$ and $\mu=mgu(L,L')$. $(D_0\lor L\lor L')\mu$ is an instance of $D'$. Thus $N \cup U\models (D_0 \lor L\lor L')\mu$.
     Since $L\mu = L'\mu$, $(D_0\lor L\lor L')\mu\models(D_0\lor L)\mu$.
     Thus $N \cup U\models (D_0 \lor L)\mu$ and $(D_0 \lor L)\mu\sigma$ is $\beta\mhyphen\mathit{false}$ since $(D_0\lor L)\mu\sigma = (D_0\lor L)\sigma$ by definition of an mgu.
    \end{itemize}

\paragraph{Equality-Resolution:}  Assume $\mathit{Equality\mhyphen Resolution}$ is applicable. Then $D= (D'\lor s \not\approx s')\sigma$ and $s\sigma=s'\sigma$, $\mu=mgu(s,s')$. Then
    \begin{itemize}
     \item \ref{appr_4:soundness}.\ref{appr_4:sound_1} - \ref{appr_4:soundness}.\ref{appr_4:sound_4} trivially hold by induction hypothesis.
     \item \ref{appr_4:soundness}.\ref{appr_4:sound_5}: By i.h. $(D'\lor s\not\approx s')\sigma$ is $\beta\mhyphen\mathit{false}$ in $\Gamma$ and $N\cup U\models (D'\lor s \not\approx s')$.
     $D'\mu$ is an instance of $(D'\lor s \not\approx s')$. Thus $(D'\lor s \not\approx s')\models D'\mu$. Thus $N\cup U\models D'\mu$. $D'\mu\sigma$ is $\beta\mhyphen\mathit{false}$ since $(D'\lor s\not\approx s')\sigma$ is $\beta\mhyphen\mathit{false}$ and $D'\mu\sigma = D'\sigma$ by definition of a mgu.
    \end{itemize}

\paragraph{Backtrack:}  Assume $Backtrack$ is applicable. Then $\Gamma = \Gamma', K, \Gamma''$ and $D=(D'\lor L)\sigma$, where $L\sigma$ is of level $k$, and $D'\sigma$ is of level $i$.
    \begin{itemize}
     \item \ref{appr_4:soundness}.\ref{appr_4:sound_1}: By i.h. $\Gamma$ is consistent. Thus $\Gamma' \subseteq \Gamma$ is consistent.
     \item \ref{appr_4:soundness}.\ref{appr_4:sound_2} - \ref{appr_4:soundness}.\ref{appr_4:sound_3}: Since $\Gamma'$ is a prefix of $\Gamma$ by i.h. this holds.
     \item \ref{appr_4:soundness}.\ref{appr_4:sound_4}: By i.h. $N\cup U \models D'\lor L$ and $N\models U$. Thus $N\models U \cup \{D'\lor L\}$
     \item \ref{appr_4:soundness}.\ref{appr_4:sound_5}: trivially holds since $D = \top$ after backtracking.
    \end{itemize}
\end{proof}

%\subsubsection{Proof of Lemma~\ref{lit_lvl_stays}}
%Assume a state $(\Gamma;N;U;\beta;k; D)$ resulting from a regular run. If $\Gamma = \Gamma_1, K^{i+1}, \Gamma_2$ and $L$ is of level $i$ in $\Gamma$ and
%$$(\Gamma; N; U; \beta; k; D)\Rightarrow^*_{\SCLEQ} (\Gamma_1, K^{i+1},\Gamma_3; N; U'; \beta; k'; D') $$
%then $L$ is of level $i$ in $\Gamma_1, K^{i+1},\Gamma_3$.
%\begin{proof}
% By induction on $(\Gamma;N;U;\beta;k;D)$:
% \begin{itemize}
%  \item Starting from the initial state, any application of a rule preserves this trivially, since there is no literal with a level on $\Gamma$.
%  \item Now assume that we have a sound state $(\Gamma;N;U;\beta;k;D)$ where this holds. We show that any application of a rule preserves it.
%  \begin{itemize}
%  \item Assume $Propagate$ or $Decide$ is applicable. Then a new literal is added to the trail. Now we have two cases. Either the literal creates one or more new cores for $L$. The levels of these cores are greater or equal to all existing cores. Thus the level of $L$ is unchanged, since the core with the lowest level specifies the level of $L$. In case that the literal does not create new cores, the level trivially does not change.
%  \item Assume $Skip$ or $Backtrack$ is applicable. Since we only skip literals of level at least $i+1$ we only remove cores of a higher level than that of $L$. Thus the level remains the same.
%  \item Assume that any other rule is applicable. In this case $\Gamma$ does not change, so it trivially holds by induction hypothesis.
%  \end{itemize}
% \end{itemize}
%\end{proof}

\subsubsection{Proof of Lemma~\ref{lem:soundcomplete:unitrewriting}}
Assume a state $(\Gamma;N;U;\beta;k;D)$ resulting from a regular run where the current level $k>0$ and a unit clause $l\approx r\in N$. Now assume a clause $C\lor L[l']_p \in N$ such that $l' = l\mu$ for some matcher $\mu$.
Now assume some arbitrary grounding substitutions $\sigma'$ for $C\lor L[l']_p$, $\sigma$ for $l\approx r$ such that $l\sigma=l'\sigma'$ and $r\sigma\prec_T l\sigma$.
Then $ (C\lor L[r\mu\sigma\sigma']_p)\sigma'\prec_{\Gamma^*} (C\lor L[l']_p)\sigma'$.
\begin{proof}
  Let $(\Gamma;N;U;\beta;k;D)$ be a state resulting from a regular run where $k> 0$ and $\Gamma = [L_1,...,L_n]$.
  Now we have two cases:
  \begin{enumerate}
  \item $\beta\prec_T (l\approx r)\sigma$. Since $(l\approx r)\sigma$ rewrites $L[l']_p\sigma'$, $\beta\prec_T L[l']_p\sigma'$ has to hold as well. Thus $(l\approx r)\sigma$ is $\beta\mhyphen\mathit{undefined}$ in $\Gamma$ and $L[l']_p\sigma'$ is $\beta\mhyphen\mathit{undefined}$ in $\Gamma$.
  By definition of a trail induced ordering $\prec_{\Gamma^*} := \prec_T$ for $\beta\mhyphen\mathit{undefined}$ literals. Thus, in case that $L[r\mu]_p)\sigma\sigma'$ is still undefined, $(L[r\mu]_p)\sigma\sigma' \prec_{\Gamma^*} (L[l']_p)\sigma'$ has to hold since $(L[r\mu]_p)\sigma\sigma' \prec_T (L[l']_p)\sigma'$. Thus, according to the definition of multiset orderings, $ (C\lor L[r\mu]_p)\sigma\sigma' \prec_{\Gamma^*} (C\lor L[l']_p)\sigma'$.
  In the case that $(L[r\mu]_p)\sigma\sigma'$ is defined, $(L[r\mu]_p)\sigma\sigma' \prec_{\Gamma^*} (L[l']_p)\sigma'$ has to hold as well by definition~\ref{appr_4:trail_ind_order}.11. Thus, according to the definition of multiset orderings, $ (C\lor L[r\mu]_p)\sigma\sigma' \prec_{\Gamma^*} (C\lor L[l']_p)\sigma'$.
  \item $(l\approx r)\sigma\prec_T \beta$.
  Since propagation is exhaustive for literals of level $0$ (cf. \ref{appr_4:def:regular-run}.2) $(l\approx r)\sigma$ is on the trail or defined and of level $0$. Now we have two cases:
    \begin{enumerate}
    \item $(L[l']_p)\sigma'$ is of level 1 or higher. Since $(L[l']_p)\sigma'$ is reducible by $(l\approx r)\sigma$, $(L[l']_p)\sigma' \neq {L_i}$ and $(L[l']_p)\sigma' \neq {comp(L_i)}$ for all $L_i\in\Gamma$. Since $(L[l']_p)\sigma'$ is of level 1 or higher, rewriting with $(l\approx r)\sigma$ does not change the defining literal of $(L[l']_p)\sigma'$. Thus $(L[r\mu]_p)\sigma\sigma' \prec_{\Gamma^*} (L[l']_p)\sigma'$ has to hold since $(L[r\mu]_p)\sigma\sigma' \prec_T (L[l']_p)\sigma'$. Thus, according to the definition of multiset orderings, $(C\lor L[r\mu]_p)\sigma\sigma' \prec_{\Gamma^*} (C\lor L[l']_p)\sigma'$
    \item $(L[l']_p)\sigma'$ is of level 0. First we show that $(L[r\mu]_p)\sigma\sigma'$ is still of level $0$. Suppose that $(L[l']_p)\sigma' = s\,\#\, s$. Then rewriting either the left or right side of the equation results in  $(L[r\mu]_p)\sigma\sigma'$. Then $core(\Gamma; (l\approx r)\sigma)$ is also a core for $(L[r\mu]_p)\sigma\sigma'$ and thus $(L[r\mu]_p)\sigma\sigma'$ must be of level $0$.
    Now suppose that $(L[r\mu]_p)\sigma\sigma' = s\,\#\, s$. Then it is of level $0$ by definition of a level.
    Finally suppose that $(L[r\mu]_p)\sigma\sigma' \not= s\,\#\, s$ and $(L[l']_p)\sigma' \not= s\,\#\, s$. Then $core(\Gamma;(L[l']_p)\sigma')\cup core(\Gamma; (l\approx r)\sigma)$ is a core for $(L[r\mu]_p)\sigma\sigma'$. Thus $(L[r\mu]_p)\sigma\sigma'$ is of level $0$.
    Since $(L[r\mu]_p)\sigma\sigma' \prec_T (L[l']_p)\sigma'$, $(L[r\mu]_p)\sigma\sigma' \prec_{\Gamma^*} (L[l']_p)\sigma'$ according to the definition of $\prec_{\Gamma^*}$. Thus, according to the definition of multiset orderings, $(C\lor L[r\mu]_p)\sigma\sigma' \prec_{\Gamma^*} (C\lor L[l']_p)\sigma'$.
  \end{enumerate}
\end{enumerate}
\end{proof}

\subsubsection{Proof of Lemma~\ref{lem:redundancy}}
Let $C,D$ be two clauses. If there exists a substitution $\sigma$ such that $C\sigma\subset D$, then $D$ is redundant with respect to $C$ and any $\prec_{\Gamma^*}$.
\begin{proof}
Let $\tau$ be a grounding substitution for $D$. Since $C\sigma\subset D$, $C\sigma\tau\subset D\tau$. Thus, for any $L\in C\sigma\tau$ it holds $L\in D\tau$ and $C\sigma\tau\not= D\tau$.
Thus, $C\sigma\tau\prec_{\Gamma^*}D\tau$ by definition of a multiset extension and $C\sigma\tau$ makes $D\tau$ redundant by definition~\ref{appr_4:prelim:def:redundancy}.
\end{proof}

%\subsubsection{Proof of Lemma~\ref{lem:subsumptionres}}
%Assume a state $(\Gamma;N;U;\beta;k;D)$ resulting from a regular run.
%Let $C\lor s\approx t$ and $D\lor L[s']$ be two clauses such that $s\sigma = s'$, $t\sigma\prec_T s\sigma$ and $C\sigma\subset D$ for some substitution $\sigma$. Then $D\lor L[s']$ is redundant with respect to $C\lor s\approx t, D\lor L[t]$ and $\prec_{\Gamma^*}$.
%\begin{proof}

%\end{proof}

\subsubsection{Auxiliary Lemmas for the Proof of Lemma~\ref{appr_4:non-red}}

\begin{lemma}
  \label{no_false_clause_after_backtrack}
  During a regular run, if $(\Gamma;N;U;\beta;k;\top)$ is the immediate result of an application of $Backtrack$, then there exists no clause $C\in N\cup U$ and a substitution $\sigma$ such that $C\sigma$ is $\beta\mhyphen\mathit{false}$ in $\Gamma$.
\end{lemma}
\begin{proof}

  We prove this by induction. For the induction start assume the state $(\Gamma';N;U\cup\{D\};\beta;i;\top)$ after the first application of $Backtrack$ in a regular run, where $D$ is the learned clause. Since $Backtrack$ was not applied before, the previous (first) application of $\mathit{Conflict}$ in a state $(\Gamma,K;N;U;\beta;k;\top)$ was immediately preceded by an application of $Propagate$ or $Decide$. By the definition of a regular run there is no clause $C\in N$ with substitution $\sigma$ such that $C\sigma$ is $\beta\mhyphen\mathit{false}$ in $\Gamma$. Otherwise $\mathit{Conflict}$ would have been applied earlier.
  By the definition of $Backtrack$, there exists no substition $\tau$ such that $D\tau$ is $\beta\mhyphen\mathit{false}$ in $\Gamma'$. Since there existed such a substitution before the application of $Backtrack$, $\Gamma'$ has to be a prefix of $\Gamma$ and $\Gamma \not=\Gamma'$. Thus there exists no clause $C\in N\cup U\cup\{D\}$ and a grounding substitution $\delta$ such that $C\delta$ is $\beta\mhyphen\mathit{false}$ in $\Gamma'$.

  For the induction step assume the state $(\Gamma';N;U\cup\{D\};\beta;i;\top)$ after $n$th application of $Backtrack$. By i.h. the previous application of $Backtrack$ did not produce any $\beta\mhyphen\mathit{false}$ clause. It follows that the the previous application of $\mathit{Conflict}$ in a state $(\Gamma,K;N;U;\beta;k;\top)$ was immediately preceded by an application of $Propagate$ or $Decide$.
  By the definition of a regular run there is no clause $C\in N\cup U$ with substitution $\sigma$ such that $C\sigma$ is $\beta\mhyphen\mathit{false}$ in $\Gamma$. Otherwise $\mathit{Conflict}$ would have been applied earlier.
  By the definition of $Backtrack$, there exists no substition $\tau$ such that $D\tau$ is $\beta\mhyphen\mathit{false}$ in $\Gamma'$. Since there existed such a substitution before the application of $Backtrack$, $\Gamma'$ has to be a prefix of $\Gamma$ and $\Gamma \not=\Gamma'$. Thus there exists no clause $C\in N\cup U\cup\{D\}$ and a grounding substitution $\delta$ such that $C\delta$ is $\beta\mhyphen\mathit{false}$ in $\Gamma'$.
\end{proof}

\begin{corollary}
  \label{confl_always_after_propdecide}
  If $\mathit{Conflict}$ is applied in a regular run, then it is immediately preceded by an application of $Propagate$ or $Decide$, except if it is applied to the initial state.
\end{corollary}

\begin{lemma}
  \label{no_decision_compl_clause}
  Assume a state $(\Gamma;N;U;\beta;k;D)$ resulting from a regular run. Then there exists no clause $(C\lor L)\in N\cup U$ and a grounding substitution $\sigma$ such that $(C\lor L)\sigma$ is $\beta\mhyphen\mathit{false}$ in $\Gamma$, $comp(L\sigma)$ is a decision literal of level $i$ in $\Gamma$ and $C\sigma$ is of level $j<i$.
\end{lemma}
\begin{proof}
  Proof is by induction. Assume the initial state $(\epsilon;N;\emptyset;\beta;0;\top)$. Then any clause $C\in N$ is undefined in $\Gamma$. Then this trivially holds.

  Now for the induction step assume a state $(\Gamma;N;U;\beta;k;D)$. Only $Propagate$, $Decide$, $Backtrack$ and $Skip$ change the trail and only $Backtrack$ adds a new literal to $U$. By i.h. there exists no clause with the above properties in $N\cup U$.

  Now assume that $Propagate$ is applied. Then a literal $L$ is added to the trail. Let $C_1\lor L_1,...,C_n\lor L_n$ be the ground clause instances that get $\beta\mhyphen\mathit{false}$ in $\Gamma$ by the application such that $L$ is the defining literal of $L_1,...,L_n$. Then $L_i$ is of level $k$ for $1\leq i\leq n$. Thus $L_i\not=comp(K)$ for the decision literal $K\in \Gamma$ of level $k$. Thus $C_1\lor L_1,...,C_n\lor L_n$ do not have the above properties.

  Now assume that $Decide$ is applied. Then a literal $L$ of level $k+1$ is added to the trail. Let $C_1\lor L_1,...,C_n\lor L_n$ be the (ground) clause instances that get $\beta\mhyphen\mathit{false}$ in $\Gamma$ by the application such that $L$ is the defining literal of $L_1,...,L_n$. By the definition of a regular run for all $L_i$ with $1\leq i\leq n$ it holds that $L_i\not=comp(L)$ or there exists another literal $K_i\in C_i$ such that $K_i$ is of level $k+1$ and $L_i\not=K_i$, since otherwise $Propagate$ must be applied.
  Thus $C_1\lor L_1,...,C_n\lor L_n$ do not have the above properties.

  Now assume that $Skip$ is applied. Then there are no new clauses that get $\beta\mhyphen\mathit{false}$ in $\Gamma$. Thus this trivially holds.

  Now assume that $Backtrack$ is applied. Then a new clause $D\lor L$ is added to $U$ and $\Gamma=\Gamma',K,\Gamma''$ such that
   there is a grounding substitution $\tau$ with $(D\lor L)\tau$ $\beta\mhyphen\mathit{false}$ in $\Gamma',K$,
   there is no grounding substitution $\delta$ with $(D\lor L)\delta$ $\beta\mhyphen\mathit{false}$ in $\Gamma'$. $\Gamma'$ is the trail resulting from the application of $Backtrack$.
   By lemma~\ref{no_false_clause_after_backtrack}, after application of $Backtrack$ there exists no clause $C\in N\cup U$ and a substitution $\sigma$ such that $C\sigma$ is $\beta\mhyphen\mathit{false}$ in $\Gamma'$. Thus there exists no clause with the above properties.
\end{proof}

\subsubsection{Proof of Lemma~\ref{appr_4:non-red}}
Let $N$ be a clause set. The clauses learned during a regular run in \SCLEQ\ are not redundant with respect to $\prec_{\Gamma^*}$ and $N \cup U$.
For the trail only
non-redundant clauses need to be considered.

We first prove that learned clauses are non-redundant and then that only
non-redundant clauses need to be considered, Lemma~\ref{lem:trailnonredcl}, below.
\begin{proof}
 Consider the following fragment of a derivation learning a clause:
 \begin{flalign}
 &\Rightarrow^{\mathit{Conflict}}_{\SCLEQ} &(\Gamma; N; U; \beta;k; D\,\cdotp \sigma) \nonumber\\
 &\Rightarrow^{\{Explore\mhyphen Refutation, Skip, Eq\mhyphen Res,Factorize\}^*}_{\SCLEQ} &(\Gamma'; N; U; \beta;l; C\,\cdotp \sigma) \nonumber\\
 &\Rightarrow^{Backtrack}_{\SCLEQ} &(\Gamma''; N; U \cup \{C\}; \beta; k'; \top) \nonumber
 \end{flalign}
 Assume there are clauses in $N' \subseteq (gnd(N\cup U)^{\preceq_{\Gamma^*} C\sigma})$ such that $N' \models C\sigma$. Since $N'\preceq_{\Gamma^*} C\sigma$ and $C\sigma$ is $\beta\mhyphen\mathit{defined}$ in $\Gamma$, there is no $\beta\mhyphen\mathit{undefined}$ literal in $N'$, as all $\beta\mhyphen\mathit{undefined}$ literals are greater than all $\beta\mhyphen\mathit{defined}$ literals. If $\Gamma \models N'$ then $\Gamma \models C\sigma$, a contradiction. Thus there is a $C'\in N'$ with $C'\preceq_{\Gamma^*} C\sigma$ such that $C'$ is $\beta\mhyphen\mathit{false}$ in $\Gamma$. Now we have two cases:
 \begin{enumerate}
  \item $\Gamma' \not = \Gamma$. Then $\Gamma = \Gamma',\Delta$.
  Thus at least one Skip was applied, so $C\sigma$ does not contain a literal that is $\beta\mhyphen\mathit{undefined}$ without the rightmost literal of $\Gamma$, therefore $C\sigma\not=D\sigma$. Suppose that this is not the case, so $C\sigma=D\sigma$. Then $D\sigma$ is $\beta\mhyphen\mathit{false}$ in $\Gamma'$.
  But since $Backtrack$ does not produce any $\beta\mhyphen\mathit{false}$ clauses by lemma \ref{no_false_clause_after_backtrack}, $\mathit{Conflict}$ could have been applied earlier on $D\sigma$ contradicting a regular run.
  Since $C' \preceq_{\Gamma^*} C\sigma$ we have that $C'\not=D\sigma$ as well.
  Thus, again since $Backtrack$ does not produce any $\beta\mhyphen\mathit{false}$ clauses by lemma \ref{no_false_clause_after_backtrack}, at a previous point in the derivation there must have been a state such that $C'$ was $\beta\mhyphen\mathit{false}$ under the current trail and $\mathit{Conflict}$ was applicable but not applied, a contradiction to the definition of a regular run.
  \item $\Gamma' = \Gamma$, then conflict was applied immediately after an application of $\mathit{Decide}$ by corollary \ref{confl_always_after_propdecide} and the definition of a regular run. Thus $\Gamma= \Delta,\allowbreak K^{(k-1):D'\cdotp\delta},\allowbreak L^{k:D''\cdotp\tau}$. $C'$ does not have any $\beta\mhyphen\mathit{undefined}$ literals.
  Suppose that $C'$ has no literals of level $k$. Then all literals in $C'$ are of level $i<k$. Since $C'$ is $\beta\mhyphen\mathit{false}$ in $\Gamma$, $C'$ is $\beta\mhyphen\mathit{false}$ in $\Delta,K$ as well, since it does not have any literals of level $k$. Thus, again since $Backtrack$ does not produce any $\beta\mhyphen\mathit{false}$ clauses by lemma \ref{no_false_clause_after_backtrack}, at a previous point in the derivation there must have been a state such that $C'$ was $\beta\mhyphen\mathit{false}$ under the current trail and $\mathit{Conflict}$ was applicable but not applied, a contradiction to the definition of a regular run.\\
  Since $C'\preceq_{\Gamma^*} C\sigma$, it may have at most one literal of level $k$, namely $comp(L)$, since $comp(L) \in C\sigma$ by definition of a regular run, since Skip was not applied, and there exists only $L$ such that $L \prec_{\Gamma^*} comp(L)$ and $L$ is of level $k$. But $L$ is $\beta\mhyphen\mathit{true}$ in $\Gamma$. Thus $L\not\in C'$ has to hold.\\
  Now suppose that $C'$ has one literal of level $k$. Thus $C' = C''\lor comp(L)$, where $C''$ is $\beta\mhyphen\mathit{false}$ in $\Delta,K$.
  But by lemma \ref{no_decision_compl_clause} there does not exist such a clause. Contradiction.\\
 \end{enumerate}
\end{proof}

\subsubsection{Auxiliary Lemma for the Proof of Lemma \ref{always_rule_applicable_D=top}}
\begin{lemma}
  \label{rewritable_below_var_lvl}
Assume a clause $L_1\lor...\lor L_m$, a trail $\Gamma$ resulting from a regular run starting from the initial state, and a reducible (by $conv(\Gamma)$) grounding substitution $\sigma$, such that $L_i\sigma$ is $\beta\mhyphen\mathit{false}$ ($\beta\mhyphen\mathit{true}$ or $\beta\mhyphen\mathit{undefined}$) in $\Gamma$ and $L_i\sigma\prec_T \beta$ for $1\leq i \leq m$.
Then there exists a substitution $\sigma'$ that is irreducible by $conv(\Gamma)$ such that $L_i\sigma'$ is $\beta\mhyphen\mathit{false}$ ($\beta\mhyphen\mathit{true}$ or $\beta\mhyphen\mathit{undefined}$) in $\Gamma$, $L_i\sigma'\prec_T \beta$ and $L_i\sigma{\downarrow_{conv(\Gamma)}} = L_i\sigma'{\downarrow_{conv(\Gamma)}}$.
\end{lemma}
\begin{proof}
Let $L_1\lor...\lor L_m$ be a clause, $\Gamma$ a trail resulting from a regular run.
Let $\sigma := \{x_1\rightarrow t_1, ..., x_n\rightarrow t_n\}$. Now set $\sigma' := \{x_1\rightarrow ({t_1}{\downarrow_{conv(\Gamma)}}), ..., x_n\rightarrow ({t_n}{\downarrow_{conv(\Gamma)}})\}$. Obviously $\sigma'$ is irreducible by $conv(\Gamma)$ and $L_i\sigma'\prec_T \beta$ for all $1\leq i \leq m$.
By definition, $conv(\Gamma)$ is a confluent and terminating rewrite system.
Since $\Gamma$ is consistent, ${t_j}{\downarrow_{conv(\Gamma)}}\approx t_j$ is $\beta\mhyphen\mathit{true}$ in $\Gamma$ for $1\leq j \leq n$. Thus there exists a chain such that $L_i\sigma\rightarrow_{conv(\Gamma)}...\rightarrow_{conv(\Gamma)}L_i\sigma'$ and $L_i\sigma'$ is $\beta\mhyphen\mathit{false}$ ($\beta\mhyphen\mathit{true}$ or $\beta\mhyphen\mathit{undefined}$) in $\Gamma$. Now there also exists a chain $L_i\sigma\rightarrow_{conv(\Gamma)}...\rightarrow_{conv(\Gamma)}L_i\sigma{\downarrow_{conv(\Gamma)}}$.
By definition of convergence there must exist a chain $L_i\sigma'\rightarrow_{conv(\Gamma)}...\rightarrow_{conv(\Gamma)}L_i\sigma{\downarrow_{conv(\Gamma)}}$. Thus $L_i\sigma{\downarrow_{conv(\Gamma)}} = L_i\sigma'{\downarrow_{conv(\Gamma)}}$.
%Assume such a clause $C$ and substitution $\sigma$ with the above properties. Now assume some $L\in C$ such that $L\sigma$ is rewritable below variable level. Since $L\sigma$ is rewritable below variable level there exists a position $p = p'p''$ such that $L|_{p'} = x$ and $L\sigma|_p = l$ where $l\approx r\in conv(\Gamma)$.
%Then we can create a new substitution $\sigma'= (\sigma\setminus \{x\rightarrow x\sigma\})\cup \{x\rightarrow x\sigma[r]_{p''}\}$. Now $L\sigma$ is false in $\Gamma$ iff $L\sigma'$ is false in $\Gamma$, since $conv(\Gamma)$ is a convergent rewrite system and $L\sigma'$ is the result of rewriting all occurrences of $l$ inside variable occurrences $x$ in $L\sigma$ by $r$.
%Furthermore for all  $L\in C$, $L\sigma$ is false in $\Gamma$ iff $L\sigma'$ is false in $\Gamma$ for the same reason.\\
%Finally we can repeat this procedure until rewriting below variable level is no longer possible. Since $conv(\Gamma)$ is terminating, this will finally happen.
\end{proof}

\subsubsection{Auxiliary Lemma for the Proof of Lemmas~\ref{form_of_stuck_states} and~\ref{theo:soundcomp:refutcomp}}
\begin{lemma}
  \label{always_rule_applicable_D=top}
  Suppose a sound state $(\Gamma; N;U;\beta;k; \top)$ resulting from a regular run. If there exists a $C\in N\cup U$ and a grounding substitution $\sigma$ such that $C\sigma$ is $\beta\mhyphen\mathit{false}$ in $\Gamma$, then $\mathit{Conflict}$ is applicable.
  Otherwise, If there exists a $C\in N\cup U$ and a grounding substitution $\sigma$ such that $C\sigma\prec_T \beta$ and there exists at least one $L\in C$ such that $L\sigma$ is $\beta\mhyphen\mathit{undefined}$,
  then one of the rules $Propagate$ or $Decide$ is applicable and a $\beta\mhyphen\mathit{undefined}$ literal $K\in D$, where $D\in gnd_{\prec_T}\beta(N\cup U)$ is $\beta\mhyphen \mathit{defined}$ after application.
\end{lemma}
\begin{proof}
  %By proposition \ref{rewritable_below_var_lvl} there always exists a irreducible substitution which retains the state of the ground clause in $\Gamma$. Thus we assume that all grounding substitutions are irreducible.\\
  Let $(\Gamma; N;U;\beta;k; \top)$ be a state resulting from a regular run.
  Suppose there exists a $C\in N\cup U$ and a grounding $\sigma$ such that $C\sigma$ is $\beta\mhyphen\mathit{false}$ in $\Gamma$, then by lemma \ref{rewritable_below_var_lvl} there exists an irreducible substitution $\sigma'$ such that $C\sigma'$ is $\beta\mhyphen\mathit{false}$. Thus $\mathit{Conflict}$ is applicable.
  Now suppose there exists a $C\in N\cup U$ and a grounding substitution $\sigma$ such that $C\sigma\prec_T \beta$ and there exists at least one $L\in C$ such that $L\sigma$ is $\beta\mhyphen\mathit{undefined}$.
  By lemma \ref{rewritable_below_var_lvl} there exists a irreducible substitution $\sigma'$ such that $L\sigma'$ is $\beta\mhyphen\mathit{undefined}$.
  Now assume that $C = C_0\lor C_1\lor L$ such that $C_1\sigma' = L\sigma'\lor...\lor L\sigma'$ and $C_0\sigma'$ is $\beta\mhyphen\mathit{false}$ in $\Gamma$. Then $\mathit{Propagate}$ is applicable. Let $C_1 = L_1,...,L_n$ and $\mu = mgu(L_1,...,L_n,L)$. Now let $[I_1,...,I_m]$ be the reduction chain application from $\Gamma$ to $L\sigma'^{k:(L\lor C_0)\mu\cdotp\sigma'}$.
  Let $I_m = (s_m\# t_m\cdotp\sigma_m, (s_m\# t_m\lor C_m)\cdotp\sigma_m, I_j, I_k, p_m)$.
  Then $L\sigma'{\downarrow_{conv(\Gamma)}} = s_m\# t_m\sigma_m$ by definition of a reduction chain application. Thus $L\sigma'$ is $\beta\mhyphen\mathit{true}$ in $\Gamma,s_m\# t_m\sigma_m$.
  Since $L\sigma{\downarrow_{conv(\Gamma)}} = L\sigma'{\downarrow_{conv(\Gamma)}}$ by lemma \ref{rewritable_below_var_lvl}, $L\sigma$ is $\beta\mhyphen\mathit{true}$ in $\Gamma,s_m\# t_m\sigma_m$ as well.
  If $C_0\sigma$ is $\beta\mhyphen\mathit{undefined}$ or $\beta\mhyphen\mathit{true}$ in $\Gamma$ then Propagate is not applicable to $C\sigma'$. If Decide is not applicable by definition of a regular run, then there exists a clause $C'\in(N\cup U)$ and a substitution $\delta$ such that Propagate is applicable. Then we can apply Propagate by definition of a regular run and a previously undefined literal gets defined after application as seen above and we are done. Now suppose that there exists no such clause. Then let $[I'_1,...,I'_l]$ be the reduction chain application from $\Gamma$ to $L\sigma'^{k+1:C\cdotp\sigma'}$ and $I'_l = (s_l\# t_l\cdotp\sigma_l, (s_l\# t_l\lor C_l)\cdotp\sigma_l, I'_j, I'_k, p_l)$.
  Then $L\sigma'{\downarrow_{conv(\Gamma)}} = (s_l\# t_l)\sigma_l$ by definition of a reduction chain application. Thus $L\sigma'$ is $\beta\mhyphen\mathit{true}$ in $\Gamma,(s_l\# t_l)\sigma_l$.
  Since $L\sigma{\downarrow_{conv(\Gamma)}} = L\sigma'{\downarrow_{conv(\Gamma)}}$ by lemma \ref{rewritable_below_var_lvl}, $L\sigma$ is $\beta\mhyphen\mathit{true}$ in $\Gamma,(s_l\# t_l)\sigma_l$ as well.
  $(s_l\# t_l)\sigma_l^{k+1:(s_l\# t_l\lor comp(s_l\# t_l))\cdotp\sigma_l}$ can be added to $\Gamma$ by definition of a regular run and also by definition of Decide since $C\in N\cup U$, $\sigma'$ is grounding for $C$ and irreducible in $conv(\Gamma)$, $L\sigma'$ is $\beta\mhyphen\mathit{undefined}$ in $\Gamma$ and $C\sigma'\prec_T \beta$.
\end{proof}

\subsubsection{Auxiliary Lemma for the Proof of Lemmas~\ref{form_of_stuck_states}}
\begin{lemma}
  \label{confl_res_at_least_lvl_1}
  Suppose a sound state $(\Gamma; N;U;\beta;k; D\cdotp\sigma)$ resulting from a regular run. Then $D\sigma$ is of level $1$ or higher.
\end{lemma}
\begin{proof}
  Let $(\Gamma; N;U;\beta;k; D\cdotp\sigma)$ be a state resulting from a regular run.
  Suppose that $D\sigma$ is not of level $1$ or higher, thus $D\sigma$ is of level $0$. Then $\mathit{Conflict}$ was applied earlier to a clause that was of level $1$ or higher.
  Thus there must have been an application of $\mathit{Explore\mhyphen Refutation}$ on a state $(\Gamma,\Gamma',L; N;U;\beta; l; D'\cdotp\sigma')$ between the state after the application of $\mathit{Conflict}$ and the current state resulting in a state $(\Gamma,\Gamma',L^{l:(L\lor C)\cdotp \delta}; N;U;\beta; l; D''\cdotp\sigma'')$ such that $D'\sigma'$ is of level $l$ and $D''\sigma''$ is of level $0$, since no other rule can reduce the level of $D'\sigma'$.
  Then there exists a $K\in D'\sigma'$ such that $L$ is the defining literal of $K$. Let $[I_1,...,I_m]$ be the refutation of $K$ and $I_j = (s_j\# t_j\cdotp\sigma_j, (s_j\# t_j\lor C_j)\cdotp\sigma_j, I_i, I_k, p_j)$ be the step that was chosen by $\mathit{Explore\mhyphen Refutation}$. Then $D''\sigma'' =(s_j\# t_j\lor C_j)\sigma_j$. $C\delta\subset C_j\sigma_j$ has to hold since $L$ is the defining literal of $K$.
  Then $C\delta$ must be of level $0$ or empty. Note that $C\delta$ is of level $l$ if $L$ is a decision literal. But then, by the definition of a regular run, $L^{l:(L\lor C)\cdotp\delta}$ must have been propagated before the first decision, since propagation is exhaustive at level $0$. Contradiction.
\end{proof}

\subsubsection{Proof of Lemma~\ref{form_of_stuck_states}}
  If a regular run (without rule Grow) ends in a stuck state $(\Gamma;N;U;\beta;k;D)$, then $D=\top$ and all ground literals $L\sigma\prec_T \beta$, where $L\lor C\in N\cup U$ are $\beta\mhyphen\mathit{defined}$ in $\Gamma$.
\begin{proof}
  First we prove that stuck states never appear during conflict resolution.
  Assume a sound state $(\Gamma;N;U;\beta;k;D\clsr\sigma)$ resulting from a regular run. Now we show that we can always apply a rule.
  Suppose that $D\sigma = (D'\lor L\lor L')\sigma$ such that $L\sigma = L'\sigma$. Then we must apply $\mathit{Factorize}$ by the definition of a regular run.
  Now suppose that $\mathit{Factorize}$ is not applicable and $\Gamma := \Gamma', L$ and $D\sigma$ is false in $\Gamma'$. If $D\sigma = (D'\lor s\not\approx s')\sigma$ such that $s\sigma = s'\sigma$, we can apply $\mathit{Equality\mhyphen Resolution}$. So suppose that $\mathit{Equality\mhyphen Resolution}$ is not applicable. Then we can apply $Skip$.
  Now suppose that $\Gamma := \Gamma', L^{k:(L\lor C)\delta}$ and $L$ is the defining literal of at least one literal in $D\sigma$, so Skip is not applicable. If $D\sigma = (D'\lor L')\sigma$ where $D'\sigma$ is of level $i<k$ and $L'\sigma$ is of level $k$ and $Skip$ was applied at least once during this conflict resolution, then Backtrack is applicable. If $Skip$ was not applied and $L=comp(L'\sigma)$ and $L$ is a decision literal, then Backtrack is also applicable.
  Otherwise, let $(s\,\#\, t)\sigma\in D\sigma$ such that $K \prec_{\Gamma^*} (s\,\#\, t)\sigma$ for all $K\in D\sigma$. $(s\,\#\, t)\sigma$ exists since Factorize is not applicable.
  By lemma~\ref{confl_res_at_least_lvl_1}, $(s\,\#\, t)\sigma$ must be of level $1$ or higher. By the definition of $\prec_{\Gamma^*}$, $L$ must be the defining literal of $(s\,\#\, t)\sigma$ since $L$ is of level $1$ or higher and any literal in $D\sigma$ that has another defining literal is smaller than $(s\,\#\, t)\sigma$.
  Now suppose that $L$ is a decision literal and $(s\,\#\, t)\sigma = comp(L)$. Then $(s\,\#\, t)\sigma$ is of level $k$ and all other literals $K\in D\sigma$ are of level $i<k$, since $(s\,\#\, t)\sigma$ is the smallest $\beta\mhyphen\mathit{false}$ literal of level $k$ and Factorize is not applicable. In this case Explore-Refutation is not applicable since a paramodulation step with the decision literal does not make the conflict clause smaller. But Backtrack is applicable in this case even if Skip was not applied earlier by the definition of a regular run. Thus $(s\,\#\, t)\sigma \not=comp(L)$ or $L$ is a propagated literal has to hold. We show that in this case Explore-Refutation is applicable.
  Let $[I_1,...,I_m]$ be a refutation of $(s\,\#\,t)\sigma$ from $\Gamma$,
  $I_m = (s_m\# t_m\cdotp\sigma_m, (s_m\# t_m\lor C_m)\cdotp\sigma_m, I_j, I_k, p_m)$. Since $[I_1,...,I_m]$ is a refutation $s_m\# t_m\sigma_m = s'\not\approx s'$. Furthermore any $I_i$ either contains a clause annotation from $\Gamma,(s\,\#\, t)\sigma^{k:D\cdotp\sigma}$ or it is a rewrite inference from $I_{j'},I_{k'}$ with $j',k'<i$.
  Thus by lemma \ref{rewrite_inference_properties} it inductively follows that $C_m\sigma_m = D'\sigma_m\lor...\lor D'\sigma_m \lor C'_1\sigma_m\lor ... \lor C'_n\sigma_m$, where $C'_1\sigma_m, ... ,C'_n\sigma_m$ are clauses from $\Gamma$ without the leading trail literal and $D\sigma=D'\sigma_m\lor (s\# t)\sigma$.
  Since $L$ is the defining literal of $(s\,\#\, t)\sigma$ there must exist at least one $C'_i$ such that $C'_i\sigma_m=C\delta$. If $L$ is a propagated literal, then any literal in $C'_i\sigma_m$ is smaller than $(s\,\#\, t)\sigma$, since they are already $\mathit{false}$ in $\Gamma'$. If $L$ is a decision literal, then $C'_i\sigma_m =comp(L)$.
  Then $comp(L)$ is smaller, since $(s\,\#\, t)\sigma \not=comp(L)$ and $(s\,\#\, t)\sigma \not= L$. Thus $comp(L)\prec_{\Gamma^*}(s\,\#\, t)\sigma$. Any other literal in $C_1\sigma_m,...,C'_n\sigma_m$ is smaller in $\prec_{\Gamma^*}$, since they are already defined in $\Gamma'$.
  Since $Factorize$ is not applicable $(s\,\#\, t)\sigma$ is also strictly maximal in $D'\sigma_m$. Thus $(s_m\# t_m\lor C_m)\sigma_m\prec_{\Gamma^*} D\sigma$ which makes $\mathit{Explore\mhyphen Refutation}$ applicable.\\
  Now by lemma \ref{always_rule_applicable_D=top} it holds that if there exists an $\beta\mhyphen\mathit{undefined}$ literal in $gnd_{\prec_T \beta}(N\cup U)$, we can always apply at least one of the rules $\mathit{Propagate}$ or $Decide$ which makes a previously $\beta\mhyphen\mathit{undefined}$ literal in $gnd_{\prec_T \beta}(N\cup U)$ $\beta\mhyphen\mathit{defined}$.
\end{proof}

\subsubsection{Proof of Lemma~\ref{backtrack_finally_applicable}}
Suppose a sound state $(\Gamma; N;U;\beta;k; D)$ resulting from a regular run where $D\not\in\{\top,\bot\}$.
If $Backtrack$ is not applicable then any set of applications of $\mathit{Explore\mhyphen Refutation}$, $Skip$, $\mathit{Factorize}$, $\mathit{Equality\mhyphen Resolution}$ will finally result in a sound state $(\Gamma'; N;U;\beta;k;D')$, where $D' \prec_{\Gamma^*} D$. Then Backtrack will be finally applicable.
\begin{proof}
  Assume a sound state $(\Gamma;N;U;\beta;k;D\clsr\sigma)$ resulting from a regular run. Let $(s\,\#\, t)\sigma\in D\sigma$ such that $L \preceq_{\Gamma^*} (s\,\#\, t)\sigma$ for all $L\in D\sigma$. If  $(s\,\#\, t)\sigma$ occurs twice in $D\sigma$, then $Factorize$ is applicable. Suppose that it is applied.
  Then $D\sigma = (D' \lor (s\,\#\, t) \lor L)\sigma$, where $L\sigma = (s\,\#\,t)\sigma$. Then $\mu = mgu(s\,\#\,t, L)$ and the new conflict clause is $(D' \lor s\,\#\, t)\mu\sigma \prec_{\Gamma^*} D\sigma$. Thus in this case we are done.
  If $Factorize$ is not applicable, then the only remaining applicable rules are $Skip$, $\mathit{Explore\mhyphen Refutation}$ and $\mathit{Equality\mhyphen Resolution}$.  If $\Gamma=\Gamma',L,\Gamma''$ where $L$ is the defining literal of $(s\,\#\, t)\sigma$, then $Skip$ is applicable $|\Gamma''|$ times, since otherwise $(s\,\#\, t)\sigma$ would not be maximal in $D\sigma$. So at some point it is no longer applicable. Since $D\sigma$ is finite, $\mathit{Equality\mhyphen Resolution}$ can be applied only finitely often.
  Thus we finally have to apply $\mathit{Explore\mhyphen Refutation}$.
  Then $[I_1,...,I_m]$ is a refutation of $(s\,\#\,t)\sigma$ from $\Gamma$, and there exists an $1\leq j\leq m$, such that
  $I_j = (s_j\# t_j\cdotp\sigma_j, (s_j\# t_j\lor C_j)\cdotp\sigma_j, I_l, I_k, p_j)$,
  $ (C_j\lor s_j\,\#\,t_j)\sigma_j \prec_{\Gamma^*} (D'\lor s\,\#\,t)\sigma$. Otherwise $\mathit{Explore\mhyphen Refutation}$ would not be applicable, contradicting lemma \ref{form_of_stuck_states}. Thus in this case we are done.\\
  Now we show that Backtrack is finally applicable. Since $\prec_{\Gamma^*}$ is well-founded and $\Gamma$ is finite there must be a state where $\mathit{Explore\mhyphen Refutation}$, $Skip$, $\mathit{Factorize}$, $\mathit{Equality\mhyphen Resolution}$ are no longer applicable. By lemma \ref{confl_res_at_least_lvl_1} the conflict clause in this state must be of level $1$ or higher, thus $\bot$ cannot be inferred.
  Suppose that it is always of level $i\geq l$ for some $l$. The smallest literal of level $l$ that is $\mathit{false}$ in $\Gamma$ is $comp(L)$, where $L$ is the decision literal of level $l$. Since we can always reduce if $\mathit{Backtrack}$ is not applicable and since we can always apply a rule by lemma \ref{form_of_stuck_states}, we must finally reach a conflict clause $comp(L)\lor C$, where $C$ is of level $j<l$. Thus $\mathit{Backtrack}$ is applicable.
\end{proof}

\subsubsection{Proof of Lemma~\ref{lem:soundcomp:termination}}
 Let $N$ be a set of clauses and $\beta$ be a ground term. Then any regular run that never uses Grow terminates.
\begin{proof}
 Assume a new ground clause $D\sigma$ is learned. By lemma~\ref{appr_4:non-red} all learned clauses are non-redundant. Thus $D\sigma$ is non-redundant. By the definition of a regular run $\mathit{Factorize}$ has precedence over all other rules. Thus $D\sigma$ does not contain any duplicate literals. By theorem~\ref{appr_4:scleq_sound}, $D\sigma\prec_T\beta$ has to hold. There are only finitely many clauses $C\sigma \prec_T\beta$, where $C\sigma$ is neither a tautology nor does it contain any duplicate literals. Thus there are only finitely many clauses $D\sigma$ that can be learned.
 Thus there are only finitely many literals that can be decided or propagated.
\end{proof}

\subsubsection{Proof of Lemma~\ref{lem:soundcomp:unsat}}
If a regular run reaches the state $(\Gamma;N;U;\beta;k;\bot)$ then $N$ is unsatisfiable.
\begin{proof}
 By definition of soundness, all learned clauses are consequences of $N\cup U$, definition \ref{appr_4:soundness}.\ref{appr_4:sound_5}, and $\Gamma$ is satisfiable, definition \ref{appr_4:soundness}.\ref{appr_4:sound_1}.
\end{proof}

\subsubsection{Proof of Theorem~\ref{theo:soundcomp:refutcomp}}
Let $N$ be an unsatisfiable clause set, and $\prec_T$ a desired term ordering.
For any ground term $\beta$ where $gnd_{\prec_T\beta}(N)$ is unsatisfiable, any regular $\SCLEQ$ run
without rule Grow will terminate by deriving $\bot$.
\begin{proof}
  Since regular runs of $\SCLEQ$ terminate we just need to prove that it terminates in a failure state.
  Assume by contradiction that we terminate in a state $(\Gamma;N;U;\beta;k;\top)$.
  If no rule can be applied in $\Gamma$ then for all $s\,\#\, t\in C$ for some arbitrary $C\in gnd_{\prec_T \beta}(N)$ it holds that $s \,\#\, t$ is $\beta\mhyphen\mathit{defined}$ in $\Gamma$ (otherwise Propagate or Decide woud be applicable, see Lemma \ref{always_rule_applicable_D=top}) and there aren't any clauses in
  $gnd_{\prec_T \beta}(N)$ $\beta\mhyphen\mathit{false}$ under $\Gamma$ (otherwise $\mathit{Conflict}$ would be applicable, see again lemma \ref{always_rule_applicable_D=top}).
  Thus, for each $C\in gnd_{\prec_T \beta}(N)$ it holds that $C$ is $\beta\mhyphen\mathit{true}$ in $\Gamma$. So we have $\Gamma \models gnd_{\prec_T \beta}(N)$, but by hypothesis there is a superposition refutation of $N$ that only uses ground literals from $gnd_{\prec_T \beta}(N)$, so also $gnd_{\prec_T \beta}(N)$ is unsatisfiable, a contradiction.
\end{proof}

\begin{lemma}[Only Non-Redundant Clauses Building the Trail] \label{lem:trailnonredcl}
  Let $\Gamma=[L_1^{i_1:C_1\cdotp\sigma_1},...,L_n^{i_n:C_n\cdotp\sigma_n}]$ be a trail.
  If $L_j^{i_j:C_j\cdotp\sigma_j}$ is a propagated literal and there exist clauses $\{D_1\lor K_1,...,D_m\lor K_m\}$ with grounding substitutions $\delta_1,...,\delta_m$
   such that $N := \{(D_1\lor K_1)\delta_1,...,(D_m\lor K_m)\delta_m\} \prec_{\Gamma^*} C_j\sigma_j$ and $\{(D_1\lor K_1)\delta_1,...,(D_m\lor K_m)\delta_m\}\models C_j\sigma_j$, then there exists a $(D_k\lor K_k)\delta_k\in N$ such that
   $$[L_1^{i_1:C_1\cdotp\sigma_1},...,L_{j-1}^{i_{j-1}:C_{j-1}\cdotp\sigma_{j-1}},K_k^{i_j:(D_k\lor K_k)\cdotp\delta_k},...,L_n^{i_n:C_n\cdotp\sigma_n}]$$ is a trail.
\end{lemma}
\begin{proof}
  Let $N = \{(D_1\lor K_1)\delta_1,...,(D_m\lor K_m)\delta_m\}$ and $L_j^{i_j:C_j\cdotp\sigma_j}$ be as above. Let $\Gamma' = [L_1^{i_1:C_1\cdotp\sigma_1},...,L_{j-1}^{i_{j-1}:C_{j-1}\cdotp\sigma_{j-1}}]$. Now suppose that for every literal $L\in N$ it holds $L \prec_{\Gamma^*} L_j$. Then every literal in $N$ is defined in $\Gamma'$ and $\Gamma' \models N$,
  otherwise $\mathit{Conflict}$ would have been applied to a clause in $N$. Thus $\Gamma'\models C_j\sigma_j$ would have to hold as well. But by definition of a trail $L_j$ is undefined in $\Gamma'$.
  Thus there must be at least one clause $(D_k\lor K_k)\delta_k\in N$ with $K_k = L_j$ and $D_k\delta_k \prec_{\Gamma^*} L_j$ (otherwise $(D_k\lor K_k)\delta_k \not\prec_{\Gamma^*} C_j\sigma_j$), such that $\Gamma' \not\models D_k$.
  Suppose that $\Gamma' \models D_k$. Then $N\not\models C_j\sigma_j$, since there exists an allocation, namely $\Gamma',\neg L_k$ such that $\Gamma',\neg L_k \models N$ but $\Gamma',\neg L_k \not\models C_j\sigma_j$. Thus we can replace $L_j^{i_j:C_j\cdotp\sigma_j}$ by $K_k^{i_j:(D_k\lor K_k)\cdotp\delta_k}$ in $\Gamma$.
\end{proof}

\subsection{Further Examples} \label{subsec:intexa}
For the following examples we assume a term ordering $\prec_{kbo}$, unique weight $1$ and with precedence $d\prec c\prec b\prec a \prec a_1 \prec...\prec a_n \prec g \prec h \prec f$. Further assume $\beta$ to be large enough.

\begin{example}[Implicit Conflict after Decision] \label{exa:impcondec}
  Consider the following clause set $N$
  \begin{center}
    $\begin{array}{c}
       C_1 := h(x) \approx g(x)\lor c\approx d \qquad C_2 := f(x) \approx g(x) \lor a \approx b\\
       C_3 := f(x) \not\approx h(x) \lor f(x) \not\approx g(x)\\
       \end{array}$
  \end{center}
  Suppose we apply the rule Decide first to $C_1$ and then to $C_2$ with substitution $\sigma = \{x\rightarrow a\}$. Then we yield a conflict with $C_3\sigma$, resulting in the following state:
  \[([h(a)\approx g(a)^{1:(h(x) \approx g(x)\lor h(x) \not\approx g(x))\cdotp \sigma},f(a)\approx g(a)^{2:(f(x) \approx g(x)\lor f(x) \not\approx g(x))\cdotp \sigma}];\]
  \[N; \{\}; 2; C_3\cdotp \sigma)\]
  According to $\prec_{\Gamma^*}$, $f(a) \not\approx h(a)$ is the greatest literal in $C_3\sigma$. Since $f(a) \approx g(a)$ is the defining literal of $f(a)\not\approx h(a)$ we can not apply Skip. Factorize is also not applicable, since $f(a)\not\approx h(a)$ and $f(a) \not\approx g(a)$ are not equal. Thus we must apply Explore-Refutation to the greatest literal $f(a) \not\approx h(a)$. The rule first creates a refutation $[I_1,...,I_5]$, where:
  \begin{center}
    $\begin{array}{ll}
        I_1 := & ((f(x)\not\approx h(x))\cdotp \sigma, C_3\cdotp \sigma,\epsilon,\epsilon,\epsilon)\\
        I_2 := & ((f(x)\approx g(x))\cdotp \sigma, (f(x) \approx g(x)\lor f(x) \not\approx g(x))\cdotp \sigma,\epsilon,\epsilon,\epsilon)\\
        I_3 := & ((h(x)\approx g(x))\cdotp \sigma, (h(x) \approx g(x)\lor h(x) \not\approx g(x))\cdotp \sigma,\epsilon,\epsilon,\epsilon)\\
        I_4 := & ((h(x)\not\approx g(x))\cdotp \sigma, (h(x) \not\approx g(x)\lor f(x) \not\approx g(x)\lor f(x) \not\approx g(x))\cdotp \sigma,I_2,I_1, 1)\\
        I_5 := & ((g(x)\not\approx g(x))\cdotp \sigma, (g(x) \not\approx g(x)\lor f(x) \not\approx g(x)\lor f(x) \not\approx g(x)\\
               & \lor h(x) \not\approx g(x))\cdotp \sigma,I_4,I_3, 1)
     \end{array}$
  \end{center}
  Explore-Refutation can now choose either $I_4$ or $I_5$. Both, $(h(x)\not\approx g(x))\sigma$ and $(g(x)\not\approx g(x))\sigma$ are smaller than $(f(x)\not\approx h(x))\sigma$ according to $\prec_{\Gamma^*}$ and false in $\Gamma$. Suppose we choose $I_5$.
  Now our new conflict state is:
  \[([h(a)\approx g(a)^{1:(h(x) \approx g(x)\lor h(x) \not\approx g(x))\cdotp \sigma},f(a)\approx g(a)^{2:(f(x) \approx g(x)\lor f(x) \not\approx g(x))\cdotp \sigma}];\]
  \[N; \{\}; 2;(g(x) \not\approx g(x)\lor f(x) \not\approx g(x)\lor f(x) \not\approx g(x) \lor h(x) \not\approx g(x))\cdotp \sigma)\]
  Now we apply $\mathit{Equality\mhyphen Resolution}$ and Factorize to get the new state
  \[([g(a)\approx h(a)^{1:(g(x) \approx h(x)\lor g(x) \not\approx h(x))\cdotp \sigma},f(a)\approx g(a)^{2:(f(x) \approx g(x)\lor f(x) \not\approx g(x))\cdotp \sigma}];\]
  \[N; \{\}; 2; (f(x) \not\approx g(x)\lor h(x) \not\approx g(x))\cdotp \sigma)\]
  Now we can backtrack. Note, that this clause is non-redundant according to our ordering, although conflict was applied immediately after decision.
\end{example}
\begin{example}
  Consider the following ground clause set $N$:
  \begin{center}
    $\begin{array}{c}
       C_1 := f(a,a)\not\approx f(b,b)\lor c\approx d \qquad C_2 := a\approx b\lor f(a,a)\approx f(b,b)\\
       \end{array}$
  \end{center}
  Suppose that we decide $f(a,a)\not\approx f(b,b)$. Then $C_2$ is false in $\Gamma$. Conflict state is as follows:
  $([f(a,a)\not\approx f(b,b)^{1:f(a,a)\not\approx f(b,b)\lor f(a,a)\approx f(b,b)}]; N; \{\}; 2; C_2)$. Explore-Refutation creates the following ground refutation for $a\approx b$, since it is greatest literal in the conflict clause:
  \begin{center}
    $\begin{array}{ll}
       I_1 := & (f(a,a)\not\approx f(b,b), f(a,a)\not\approx f(b,b)\lor f(a,a)\approx f(b,b),\epsilon,\epsilon,\epsilon)\\
       I_2 := & (a\approx b,C_2,\epsilon,\epsilon,\epsilon)\\
       I_3 := & (f(b,a)\not\approx f(b,b), f(b,a)\not\approx f(b,b)\lor f(a,a)\approx f(b,b)\lor f(a,a)\approx f(b,b),\\
             & I_2,I_1, 11)\\
       I_4 := & (f(b,b)\not\approx f(b,b), f(b,b)\not\approx f(b,b)\lor f(a,a)\approx f(b,b)\lor f(a,a)\approx f(b,b)\\
             & \lor f(a,a)\approx f(b,b),I_3,I_1, 12)
     \end{array}$
  \end{center}
  As one can see, the intermediate result $f(b,a)\not\approx f(b,b)$ is not false in $\Gamma$. Thus it is no candidate for the new conflict clause. We have to choose $I_4$. The new state is thus:
  $$([f(a,a)\not\approx f(b,b)^{1:f(a,a)\not\approx f(b,b)\lor f(a,a)\approx f(b,b)}]; N; \{\}; 2;$$
  $$f(b,b)\not\approx f(b,b)\lor f(a,a)\approx f(b,b)\lor f(a,a)\approx f(b,b)\lor f(a,a)\approx f(b,b))$$
  Now we can apply Equality-Resolution and two times Factorize to get the final clause $f(a,a)\approx f(b,b)$ with which we can backtrack.
\end{example}

\begin{example}
Assume clauses:
  \begin{center}
    $\begin{array}{ll}
       C_1     & := b\approx c \lor c\approx d\\
       C_2     & := a_1\approx b\lor a_1\approx c\\
               & \qquad ...\\
       C_{n+1} & := a_n\approx b\lor a_n\approx c
       \end{array}$
  \end{center}
  The completeness proof of superposition requires that adding a new literal to an interpretation does not make any smaller literal true. In this example, however, after adding $b\approx c$ to the interpretation, we cannot any further literal, since it breaks this invariant. So in superposition we would have to add the following clauses with the help of the $Equality Factoring$ rule:
  \begin{center}
    $\begin{array}{ll}
       C_{n+2}  & := b\not\approx c\lor a_1\approx c\\
                & \qquad ...\\
       C_{2n+1} & := b\not\approx c\lor a_n\approx c\\
       \end{array}$
  \end{center}
  In \SCLEQ~on the other hand we can just decide a literal in each clause to get a model for this clause set. As we support undefined literals we do not have to bother with this problem at all. For example if we add $b\approx c$ to our model, both literals $a_1\approx b$ and $a_1\approx c$ are undefined in our model. Thus we need to decide one of these literals to add it to our model.
\end{example}

\begin{example}[Rewriting below variable level] \label{exa:rewvarlev}
  Assume the clause set $N$:
  \begin{center}
    $\begin{array}{c}
       C_1 := f(x)\approx h(b)\lor x\not\approx g(a) \qquad C_2 := c\approx d \lor f(g(b))\not\approx h(b)\\
       C_3 := a\approx b \lor f(g(b)) \approx h(b)\\
       \end{array}$
  \end{center}
  Let $\sigma =\{x\rightarrow g(a)\}$. $C_1\sigma$ must be propagated: $\Gamma = [f(g(a))\approx h(b)^{0:C_1\sigma}]$. Now suppose that we decide $f(g(b))\not\approx h(b)$. Then $\Gamma = [f(g(a))\approx h(b)^{0:C_1\sigma},f(g(b))\not\approx h(b)^{1:f(g(b))\not\approx h(b)\lor f(g(b))\approx h(b)}]$ and $C_3$ is a conflict clause. Explore-Refutation now creates the following refutation for $a\approx b$:
  \begin{center}
    $\begin{array}{ll}
        I_1 := & (f(x)\approx h(b)\cdotp\sigma, C_1\cdotp\sigma,\epsilon,\epsilon,\epsilon)\\
        I_2 := & (f(g(b))\not\approx h(b), C_2,\epsilon,\epsilon,\epsilon)\\
        I_3 := & (a\approx b, C_3,\epsilon,\epsilon,\epsilon)\\
        I_4 := & (f(g(b))\approx h(b),f(g(b))\approx h(b)\lor g(a)\not\approx g(a) \lor f(g(b)) \approx h(b),I_3,I_1,\epsilon)\\
        I_5 := & (h(b)\not\approx h(b),h(b)\not\approx h(b)\lor g(a)\not\approx g(a) \lor f(g(b)) \approx h(b)\\
               & \lor f(g(b)) \approx h(b), I_4,I_2,\epsilon)
       \end{array}$
  \end{center}
  Multiple applications of Equality-Resolution and Factorize result in the final conflict clause $C_4 := f(g(b)) \approx h(b)$ with which we can backtrack. The clause set resulting from this new clause is:
  \begin{center}
    $\begin{array}{c}
       C_1 = f(x)\approx h(b)\lor x\not\approx g(a) \qquad C'_2 = c\approx d\\
       C_4 = f(g(b))\approx h(b)\\
       \end{array}$
  \end{center}
   where $C'_2$ is the result of a unit reduction between $C_4$ and $C_2$.
   Note that the refutation required rewriting below variable level in step $I_4$. Superposition would create the following clauses (Equality-Resolution and Factorization steps are implicitly done):
   \begin{center}
   $\begin{array}{ll}
      N & \Rightarrow_{Sup(C_2,C_3)} N_1\cup \{C_4 := c\approx d \lor a\approx b\}\\
        & \Rightarrow_{Sup(C_1,C_2)} N_2\cup \{C_5 := c\approx d\lor g(a)\not\approx g(b)\}\\
        & \Rightarrow_{Sup(C_4,C_5)} N_3\cup \{C_6 := c\approx d\}\\
      \end{array}$
 \end{center}
  For superposition the resulting clause set is thus:
  \begin{center}
    $\begin{array}{c}
       C_1 = f(x)\approx h(b)\lor x\not\approx g(a) \qquad C_2 = a\approx b\lor f(g(b))\approx h(b)\\
       C_6 = c\approx d\\
       \end{array}$
  \end{center}
\end{example}

\begin{example}[Propagate Smaller Equation] \label{exa:propsmeq}
 Assume the ground clause set $N$ solely built out of constants
  \begin{center}
    $\begin{array}{c}
       C_1 := c\approx d \qquad C_2 := c\not\approx d \lor a\approx b\\
       C_3 := a\not\approx b \lor a \approx c\\
       \end{array}$
  \end{center}
\end{example}
and the trail $\Gamma := [c\approx d^{0:C_1}, a\approx b^{0:C_2}, b\approx d^{0:C_3}]$. Now, although the first two steps propagated equations
that are strictly maximal in the ordering in their respective clauses, the finally propagated equation $b\approx d$ is smaller in the term ordering $\prec_{kbo}$
than $a\approx b$. Thus the structure of the clause set forces propagation of a smaller equation in the term ordering. So
the more complicated trail ordering is a result of following the structure of the clause set rather than employing an a priori fixed ordering.
\end{document}